\documentclass{article}

\usepackage{arxiv}

\usepackage[utf8]{inputenc} 
\usepackage[T1]{fontenc}    
\usepackage{hyperref}       
\usepackage{url}            
\usepackage{booktabs}       
\usepackage{amsfonts}       
\usepackage{nicefrac}       
\usepackage{microtype}      
\usepackage{lipsum}		
\usepackage{graphicx}
\usepackage{multirow}

\usepackage[export]{adjustbox}

\usepackage[english]{babel}

\usepackage{doi}

\usepackage{comment}  

\usepackage{amsmath} 

\usepackage{amssymb} 

\usepackage{amsthm}

\usepackage[ruled,linesnumbered]{algorithm2e}

\usepackage{float}

\usepackage{subcaption}

\usepackage{silence}

\usepackage{dirtytalk}

\newtheorem{theorem}{Theorem}[section]
\newtheorem{lemma}[theorem]{Lemma}

\title{Prioritized-MVBA: A New Approach to Design an Optimal Asynchronous Byzantine Agreement Protocol}

\author{ {Nasit S Sony} \\
	University of California, Merced\\
	CA 95340, USA \\
	\texttt{nsony@ucmerced.edu} \\
	\And
	{Xianzhong Ding} \\
	Lawrence Berkeley National Laboratory\\
	CA 94720, USA \\
	\texttt{dingxianzhong@lbl.gov} \\
}

\renewcommand{\shorttitle}{\textit{Prioritized-MVBA} }

\hypersetup{
pdftitle={Prioritized-MVBA: A New Approach to Design an Optimal Asynchronous Byzantine Agreement Protocol},
pdfsubject={q-bio.NC, q-bio.QM},
pdfauthor={Nasit S Sony, Xianzhong Ding, Mukesh Singhal},
pdfkeywords={Distributed Systems, Blockchain, Consensus Protocols},
}

\begin{document}
\maketitle

\begin{abstract}
Multi-valued Byzantine agreement (MVBA) protocols are critical components in designing atomic broadcast and fault-tolerant state machine replication protocols in asynchronous networks. While these protocols have seen significant advancements, challenges remain in optimizing their communication and computation efficiency without sacrificing performance. In this paper, we address the challenge of achieving agreement in MVBA without incurring extra computation and communication rounds. Our approach leverages an analysis of message distribution patterns in asynchronous networks, observing that a subset of $f+1$ parties (where at least one party is honest), can achieve an agreement more efficiently than relying on all $n$ parties, where $n=3f+1$, $f$ is the maximum number of faulty parties. We introduce a novel protocol, Prioritized-MVBA (pMVBA), which integrates a committee-based selection process and the asynchronous binary Byzantine agreement (ABBA) protocol. In this design, a randomly selected subset of $f+1$ parties broadcast their requests, collect verifiable proofs, and utilize these proofs within the ABBA framework to reach an agreement. The proposed pMVBA protocol is resilient to up to $\lfloor \frac{n}{3} \rfloor$ Byzantine failures and achieves optimal performance, with an expected runtime of $O(1)$, message complexity of $O(n^2)$, and communication complexity of $O((l+\lambda)n^2)$ (similar to state-of-the-art protocols), where $n$ is the number of parties, $l$ is the input bit length, and $\lambda$ is the security parameter. The protocol attains $68\%$ higher throughput and $18\%$ lower latency than the classic MVBA protocol (the most efficient protocol).
\end{abstract}


\section{Introduction}
\label{sec:intro}

Byzantine Agreement (BA) protocols are fundamental to implementing decentralized infrastructures, particularly in the context of blockchain and decentralized applications. These protocols ensure agreement among parties, even in the presence of malicious actors, making them crucial for maintaining the integrity and functionality of distributed systems. The renewed interest in BA protocols has been driven by the success of Bitcoin \cite{SATOSHI08} and other decentralized applications \cite{SC}, which rely heavily on these protocols for achieving agreement \cite{CACHIN17}. Despite the existence of practical BA protocols \cite{ABRAHAM20,BYZ08,HONEYBADGER01,PASS17,SHRESTHA20,Hotstuff01}, their reliance on time parameters for ensuring liveness poses challenges in designing agreement protocols. The FLP impossibility result \cite{CONS03} demonstrates that deterministic BA protocols cannot guarantee agreement in asynchronous networks. This necessitates the development of efficient asynchronous Byzantine Agreement protocols \cite{BYZ17,DUAN18,FASTERDUMBO,HONEYBADGER01,LU20,SPEEDINGDUMBO,Joleton,BOLTDUMBO,DAG} that can overcome the limitations of existing approaches, particularly in terms of communication complexity.

Achieving agreement in asynchronous networks efficiently remains a significant challenge. Classic Multi-Valued Byzantine Agreement (MVBA) protocol like Cachin's \cite{CACHIN01}, while effective, suffer from high communication complexity, often involving terms like \(O(n^3)\). This high complexity is impractical for large-scale systems, necessitating new approaches to reduce communication overhead while ensuring optimal resilience and correctness. Cachin et al. \cite{CACHIN01} introduced the concept of MVBA, where nodes agree on a value from a large input space rather than a binary decision. Cachin’s MVBA protocol provides polynomial-time complexity but suffers from high communication overhead, with a complexity of \(O(n^3)\). VABA \cite{BYZ17} introduced a view-based approach to eliminate the \(O(n^3)\) term, while Dumbo-MVBA \cite{LU20} used erasure codes to handle large input sizes more efficiently, reducing the communication complexity to \(O(ln + \lambda n^2)\). However, these optimizations often introduced additional rounds of communication and computation.

In this context, we propose a novel approach to optimizing MVBA protocols by introducing the Prioritized Multi-Valued Byzantine Agreement (pMVBA) protocol. Our goal is to achieve agreement in each instance without the need for extra computation and communication rounds, maintaining optimal performance metrics. The key innovation of our approach is the integration of a committee-based method within the classical MVBA framework, where a randomly selected subset of parties, rather than the entire set of parties, broadcast proposals.

The pMVBA protocol leverages the concept of a committee, wherein a subset of \(f + 1\) parties is chosen to broadcast their proposals. This approach ensures that at least one honest party is always included in the subset, enhancing the protocol’s resilience. By dynamically selecting parties for each protocol instance, we mitigate the risk of adversarial attacks targeting specific nodes. This method significantly reduces the number of broadcast messages and the associated communication overhead. Central to our protocol is the integration of the Asynchronous Binary Byzantine Agreement (ABBA) protocol, which facilitates reaching an agreement on one of the proposals broadcast by the committee members. The ABBA protocol ensures that the selected parties’ proposals are agreed upon with probability $1$, maintaining the integrity and functionality of the distributed systems.

Our pMVBA protocol achieves several key improvements over existing MVBA protocols while exhibiting optimal resilience against Byzantine failures, with an expected runtime of \(O(1)\), optimal message complexity of \(O(n^2)\), and optimal communication complexity of \(O((l + \lambda)n^2)\), where \(n\) is the number of parties, \(l\) is the bit length of the input, and \(\lambda\) is the security parameter. The key improvements include removing the need for multiple instances of the protocol or extra rounds of messages and cryptographic computation to reach an agreement on a party's request while keeping the optimal communication complexity. These enhancements make the pMVBA protocol suitable for large-scale decentralized applications, addressing the scalability challenges faced by traditional MVBA protocols.
\begin{figure*}[h!]
     \centering
     \begin{subfigure}[b]{0.8\textwidth}
         \centering
         \includegraphics[width=\textwidth]{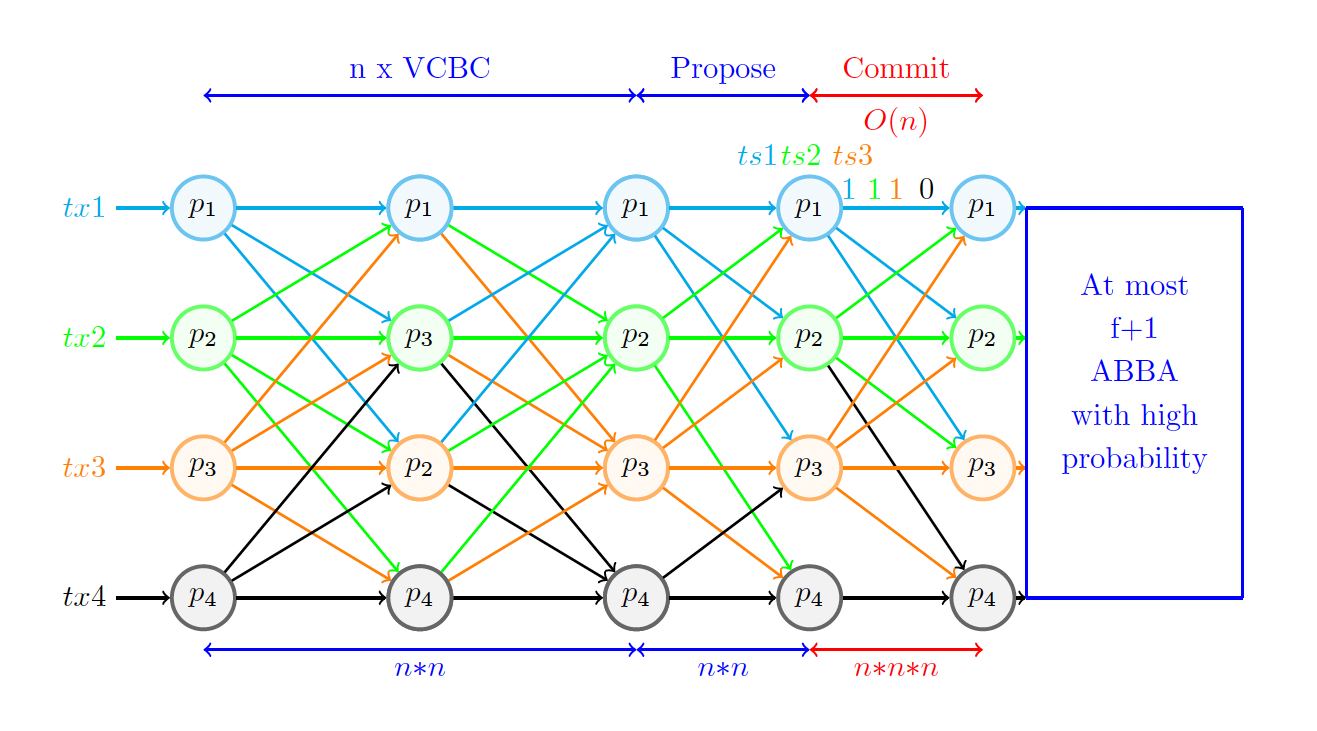}
         \caption{O(n) bottleneck}
         \label{fig:classicMVBA}
     \end{subfigure}
     
     \vfill
     \begin{subfigure}[b]{0.78\textwidth}
         \centering
         \includegraphics[width=\textwidth]{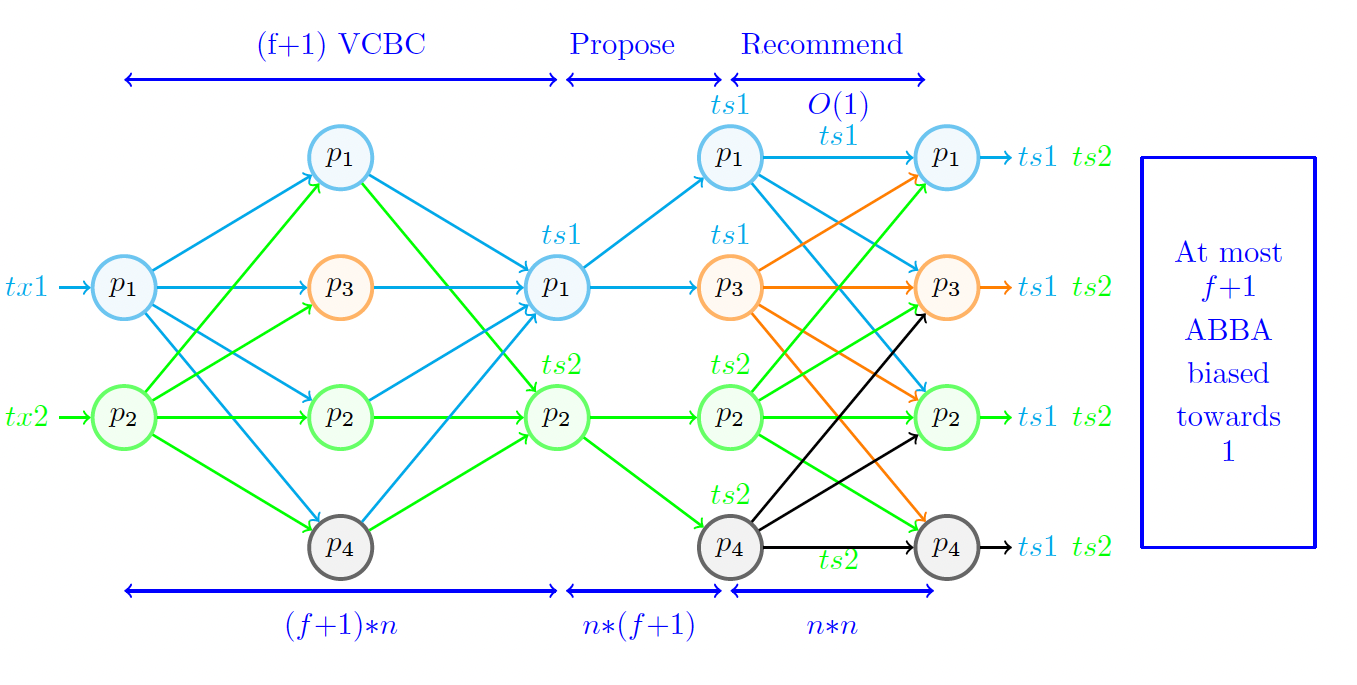}
         \caption{ No bottleneck}
         \label{fig:proposedMVBA}
     \end{subfigure}
     \caption{Comparison of two protocols. We show the time complexity of each sub-protocol and the step where our protocol wins in terms of communication. The $tx$ represents the transaction, the $ts$ represents the threshold signature. $1,1,1,0$ represents the bit array and each color for the particular party and its messages. }
     \label{fig:Comparison}
   \end{figure*}

The main contributions of this paper are as follows:
The main contributions of this project are as follows:
\begin{itemize}

\item We propose the pMVBA protocol, which reduces the communication complexity from \(O(n^3)\) \cite{SECURE03} to \(O((l + \lambda)n^2)\) by dynamically selecting a subset of parties using the cryptographic coin-tossing technique to broadcast their proposals.

\item We introduce a committee in the classical MVBA protocol \cite{SECURE03} that ensures that the protocol reaches an agreement with probability $1$, unlike other committee-based protocols. We replace the $O(n)$ bit \textit{Commit} step of the classic protocol with $O(1)$ \textit{Recommend} step.

\item We integrate the ABBA biased towards $1$ protocol to mitigate the effect of the reduction in the total number of proposals due to the introduction of the committee.

\item We present proof, extensive case studies, and experimental results to show that our approach (\textit{committee} and \textit{Recommend} step) fills the requirement of the ABBA protocol to reach an agreement on a committee member's proposal.  



\item We implement pMVBA and extensively test it among $ n = 4$, $7$ or $10$ nodes in gRPC based simulator, making detailed
Comparison to the state-of-the-art MVBA protocols, including  ClassicMVBA\cite{SECURE03} (CKSP01), and one recent result, VABA (Abraham et al. \cite{BYZ17}). At all system scales, the peak throughput of pMVBA is better than any of the other tested MVBA protocols (up to $180\%$). (see Table \ref{table:Table_INTRO}).

More importantly, the latency of pMVBA protocol is significantly
less than that of others (e. g. VABA \cite{BYZ17}, CKSP01 \cite{SECURE03}). Actually, the latency of pMVBA is nearly independent of its throughput, indicating the effectiveness of the protocol to resolve the throughput-latency tension lying in the previous protocols. As shown in Figure \ref{fig:LatencyIncrememnt}, we see that pMVBA has the lowest increment in terms of latency, and the increment ratio is also lowest among other protocols. This result indicates the effectiveness of the proposed MVBA protocol.

\begin{table}[htbp]
    
\caption{ In comparison with existing MVBA protocols.}
\begin{center}
\begin{tabular}{|c c c c c c c |} 
 \hline
\multirow{2}{4em}{Protocols} & \multicolumn{3}{ c }{Peak throughput (tps)} & \multicolumn{3}{ c |}{Latency at peak-tps(sec)}\\
& n=4 & n=7 & n = 10 & n=4 & n=7 & n = 10\\
\hline
VABA \cite{BYZ17} & 30.5  & 15.1  & 9 & 16.39 & 33 &55.35 \\
\hline
CKSP01 \cite{SECURE03} & 53.47  & 23.8 & 12 & 9.35 & 21 & 41.92\\
\hline
pMVBA (ours)  & 85.67 & 37.87  & 22 & 5.84 & 13.2 & 22.89\\
\hline
\end{tabular}
\label{table:Table_INTRO}
\end{center}
\end{table}

\begin{figure*}[htbp]
   \centering
         \includegraphics[width=0.7\textwidth]{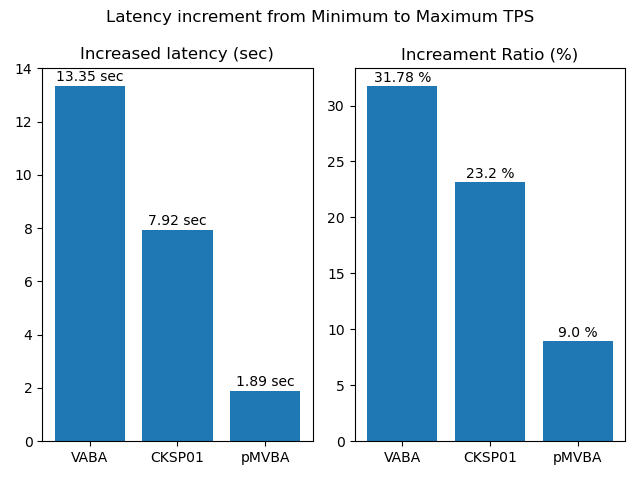}
         \caption{Latency increment of MVBA protocols  when throughput increases from minimum to maximum (n=10)}
         \label{fig:LatencyIncrememnt} 
\end{figure*}
\end{itemize}

\section{Related Work}
\label{sec:related}

This section provides an overview of the Byzantine Agreement protocols and the key developments in the Byzantine Agreement protocols, focusing on four main areas: Byzantine Agreement Protocols in partially-synchronous network, MVBA Protocols in asynchronous network, Committee-Based Protocols, and Optimistic Fastlane Mechanisms to combine both the partially-synchronous network and the asynchronous network. 

\noindent\textbf{Byzantine Agreement Protocols.}
Byzantine Agreement (BA) protocols have long been a cornerstone of fault-tolerant distributed computing. The Byzantine Generals Problem, introduced by Lamport et al. \cite{LAMPORT82}, laid the theoretical foundation for BA protocols, emphasizing the challenges of achieving an agreement in the presence of faulty or malicious nodes. Early BA protocols, such as those by Pease et al. \cite{PEASE80}, focused on synchronous networks where communication occurs in fixed rounds. However, these protocols often relied on time assumptions for ensuring liveness.

Practical Byzantine Fault Tolerance (PBFT), proposed by Castro and Liskov \cite{CASTRO99}, was a significant advancement in making BA protocols practical for real-world applications. PBFT introduced an efficient consensus protocol for partially-synchronous networks, ensuring safety and liveness even with up to \( f < \frac{n}{3} \) Byzantine faults. Several modern BA protocols build upon the principles of PBFT, including HotStuff \cite{HOTSTUFF19} and Tendermint \cite{TENDERMINT14}, optimizing various aspects such as communication complexity and responsiveness.

The FLP impossibility result \cite{FLP85} demonstrated that no deterministic BA protocol could guarantee an agreement in an asynchronous network with a single faulty node. This led to the exploration of probabilistic and randomized approaches for ABA. Ben-Or’s protocol \cite{BENOR83} was one of the earliest randomized BA protocols, using coin-flipping to break symmetry among nodes. Recent protocols like HoneyBadgerBFT \cite{MILLER16} and Dumbo \cite{DUMBO20} have focused on optimizing ABA for high throughput and low latency.

\noindent\textbf{Multi-Valued Byzantine Agreement Protocols.}
Cachin et al. \cite{CACHIN01} introduced the concept of MVBA, where nodes agree on a value from a large input space rather than a binary decision. Cachin’s MVBA protocol provides polynomial-time complexity but suffers from high communication overhead, with a complexity of \( O(ln^2 + \lambda n^2 + n^3) \).

Subsequent works, such as VABA \cite{BYZ17} and Dumbo-MVBA \cite{LU20}, aimed to reduce the communication complexity of MVBA. VABA introduced a view-based approach to eliminate the \( O(n^3) \) term, while Dumbo-MVBA used erasure codes to handle large input sizes more efficiently, reducing the communication complexity to \( O(ln + \lambda n^2) \). However, these optimizations often introduced additional rounds of communication and expensive cryptograhic computations.

Our focus on optimizing the classical MVBA protocol aims to address these challenges by reducing the number of broadcasting parties, thus lowering communication overhead while maintaining correctness and resilience.

\noindent\textbf{Committee-Based Protocols.}
Committee-based protocols have emerged as a promising approach to reduce communication complexity in BA protocols. Algorand \cite{BYZ21} introduces the concept of committee members, wherein a predefined set of parties collaborates to reach an agreement. Algorand operates within a permission-less (non-trusted) setting. In a similar vein, COINcidence \cite{COINCIDENCE} presents a protocol within a trusted setup, leveraging committee-based participation to achieve agreement with high probability.

Our approach diverges from both Algorand and COINcidence. Unlike Algorand, we operate within a trusted setup, and unlike both Algorand and COINcidence, our protocol involves selected parties undertaking a singular step. Furthermore, our protocol guarantees an agreement on a single party's requests with a probability of 1. Conversely, FasterDumbo \cite{FASTERDUMBO} utilizes the committee concept to minimize the number of components in the agreement process, yielding outputs from \( n-f \) parties' requests. Beat \cite{DUAN18} amalgamates the most effective elements of HoneyBadgerBFT's concrete implementation to showcase superior performance across varied settings.

\noindent\textbf{Optimistic Fastlane Mechanisms.}
Certain works, such as those by Victor et al. \cite{VICTOR01} and Cachin et al. \cite{CACHIN02}, have explored the incorporation of an optimistic 'fastlane' mechanism, designed for operation within partially synchronous networks. However, in instances of sluggish agreement processes, these mechanisms transition to asynchronous Byzantine agreement protocols. Recent endeavors, such as those by Joleton \cite{Joleton} and Boltdumbo \cite{BOLTDUMBO}, have adopted a similar strategy. Our work remains pertinent within this context as we rely on asynchronous protocols when the network exhibits Byzantine behavior, particularly within the pessimistic path.

Given our protocol's focus on outputting requests from a single party, we omit comparisons with DAG-rider \cite{DAG} and Aleph \cite{ALEPH}, which utilize directed acyclic graphs for agreement alongside sequential ACS. Additionally, recent efforts have aimed to eliminate the private setup phase in asynchronous Byzantine agreement protocols, emphasizing asynchronous distributed key generation \cite{DIS01,DIS02,DIS03}.

\section{System Model}
We use the standard notion \cite{SECURE02,CACHIN01} to describe a distributed algorithm involving \( n \) parties and an adversary in an authenticated setting. This section explains the foundational assumptions, network model, adversary model, computational model, and cryptographic tools used in our protocol.

\subsection{System and Network Assumptions}
\noindent\textbf{Parties and Their Setup:}
There are \( n \) parties, denoted as \( p_1, p_2, \ldots, p_n \). We consider a trusted setup of threshold cryptosystems. The system provides a secret key and a public key to each party before the protocol starts. A party uses its secret key to sign a message, and other parties use the corresponding public key to verify the signed message. The generation and distribution of these keys are out of the scope of this paper. We follow standard literature for key generation and distribution and refer interested readers to \cite{BORN01,BOLD01,Kate01}. We use the node and party alternatively throughout the paper.

\noindent\textbf{Network Model:}
Parties are connected via reliable, authenticated point-to-point channels. Here, reliable implies that if an honest party \( p_i \) sends a message to another honest party \( p_j \), the adversary cannot modify or drop the message. However, the adversary can delay the message delivery to influence the protocol execution time. Since we consider an adaptive adversary that can corrupt a party at any time during the protocol execution, the adversary can corrupt a party \( p_i \) after it sends a message and then make the party \( p_i \) drop the message.

\noindent\textbf{Adversary Model:}
We consider an adaptive adversary that can corrupt any party during the protocol execution. If the adversary takes control of a party, the party is corrupted, reveals its internal state to the adversary, behaves arbitrarily, and remains corrupted. An honest party follows the protocol, keeps its internal state secret from the adversary, and remains uncorrupted throughout the protocol execution. The adversary can corrupt \( f \) parties among the \( n \) parties, where \( f < \frac{n}{3} \).

\noindent\textbf{Computational Model:}
We adopt standard modern cryptographic assumptions and definitions from \cite{SECURE02,CACHIN01}. The assumptions allow the parties and the adversary to be probabilistic polynomial-time interactive Turing machines. This means that upon receiving a message, a party carries out some computations, changes its state, generates outgoing messages if required, and waits for the next incoming message. To rule out infinite protocol executions and restrict the run time of the adversary, we require the message bits generated by honest parties to be probabilistic uniformly bounded by a polynomial in the security parameter \( \lambda \). Therefore, we assume that the number of parties \( n \) is bounded by a polynomial in \( \lambda \).

\subsection{Validated Asynchronous Byzantine Agreement:}
Multi-valued Byzantine Agreement (MVBA) allows parties to agree on an arbitrary string \( \{0,1\}^l \), which must satisfy a predefined validity condition before the honest parties agree on it. This protocol ensures that the agreed-upon values are valid even when inputs come from malicious parties. This property is crucial for ensuring the protocol's correctness. Therefore, it is called the validated asynchronous byzantine agreement \cite{CACHIN01, VABA, BYZ20}.
\newline
\noindent{\textbf{Multi-valued Byzantine Agreement}} \label{MVBAD}
The MVBA protocol provides a polynomial-time computable predicate $Q$. Each party inputs a value $v$ that must satisfy the condition specified by $Q$, which is application-dependent. The MVBA protocol guarantees the following properties with negligible probability of failure:

\begin{itemize}
    \item \textbf{Liveness.} If all honest parties are activated, then all honest parties will reach a decision.
    \item \textbf{External Validity.} If an honest party outputs a value $v$, then $v$ satisfies the predicate $Q(v) = \text{true}$.
    \item \textbf{Agreement.} If two honest parties output values $v$ and $v'$, respectively, then $v = v'$.
    \item \textbf{Integrity.} If honest parties output a value $v$, then $v$ must have been input by a party.
    \item \textbf{Efficiency.} The number of messages generated by honest parties is uniformly bounded with high probability.
\end{itemize}

\subsection{Preliminaries}
\noindent\textbf{Asynchronous Binary Byzantine Agreement Biased Towards 1:}
The ABBA protocol allows parties to agree on a single bit \( b \in \{0,1\} \).

The ABBA protocol guarantees the following properties. Additionally, the biased external validity property applies to the biased ABBA protocol.

\begin{itemize}
    \item \textbf{Agreement.} If an honest party outputs a bit $b$, then every honest party outputs the same bit $b$.
    \item \textbf{Termination.} If all honest parties receive input, then all honest parties will output a bit $b$.
    \item \textbf{Validity.} If any honest party outputs a bit $b$, then $b$ was the input of at least one honest party.
    \item \textbf{Biased External Validity.} If at least $\langle f + 1 \rangle$ honest parties propose $1$, then any honest party that terminates will decide on $1$.
\end{itemize}

\noindent\textbf{(1, \(\kappa\), \(\epsilon\))- Committee Selection:}
The committee selection protocol is executed among \( n \) parties (identified from 1 through \( n \)). This protocol ensures that an honest party outputs a \(\kappa\)-size committee set \( C \) with at least one honest member, given that at least \( f+1 \) honest parties participate. A protocol is a $(1, \kappa, \epsilon)$-Committee Selection protocol if it satisfies the following properties with negligible probability of failure in the cryptographic security parameter $\lambda$:

\begin{itemize}
    \item \textbf{Termination.} If $\langle f+1 \rangle$ honest parties participate in the committee selection and the adversary delivers the messages, then the honest parties will output the set $C$.
    \item \textbf{Agreement.} Any two honest parties will output the same set $C$.
    \item \textbf{Validity.} If any honest party outputs the set $C$, then: 
        \begin{itemize}
            \item (i) $|C| = \kappa$,
            \item (ii) the probability of every party $p_i$ being in $C$ is the same, and 
            \item (iii) $C$ contains at least one honest party with probability $1-\epsilon$.
        \end{itemize}
    \item \textbf{Unpredictability.} The probability of the adversary predicting the committee $C$ before any honest party participates is at most $\frac{1}{^nC_\kappa}$.
\end{itemize}

Guo et al. \cite{FASTERDUMBO} constructed the $(1, \kappa, \epsilon)$-Committee Selection protocol using a threshold coin-tossing mechanism (see Appendix \ref{TCT}), which is derived from threshold signatures. The protocol ensures that at least one honest party is a committee member with overwhelming probability $1-\epsilon - \text{neg}(\lambda)$, where $\text{neg}(\lambda)$ is a negligible function in the cryptographic security parameter $\lambda$, and $\epsilon$ is $exp(-\Omega\kappa)$.

\noindent\textbf{Cryptographic Abstractions:}
We design a distributed algorithm in authenticated settings where we use robust, non-interactive threshold signatures to authenticate messages and a threshold coin-tossing protocol to select parties randomly \cite{SECURE02}.

\begin{enumerate}
    \item \textbf{Threshold Signature Scheme:}
    We utilize a threshold signature scheme introduced in \cite{THRESH01,SECURE02}. The basic idea is that there are \( n \) parties, up to \( f \) of which may be corrupted. The parties hold shares of the secret key of a signature scheme and may generate shares of signatures on individual messages. \( t \) signature shares are both necessary and sufficient to construct a signature where \( f < t \leq (n-f) \). The threshold signature scheme also provides a public key \( pk \) along with secret key shares \( sk_1, \ldots, sk_n \), a global verification key \( vk \), and local verification keys \( vk_1, \ldots, vk_n \). Initially, a party \( p_i \) has information on the public key \( vk \), a secret key \( sk_i \), and the verification keys for all the parties' secret keys. We describe the security properties of the scheme and related algorithms in Appendix \ref{TSS}.
    
    \item \textbf{Threshold Coin-Tossing Scheme:}
    The threshold coin-tossing scheme, as introduced in \cite{THRESH01,SECURE02}, involves parties holding shares of a pseudorandom function \( F \) that maps the name \( C \) (an arbitrary bit string) of a coin. We use a distributed pseudorandom function as a coin that produces \( k'' \) random bits simultaneously. The name \( C \) is necessary and sufficient to construct the value \( F(C) \in \{0,1\}^{k''} \) of the particular coin. The parties may generate shares of a coin — \( t \) coin shares are both necessary and sufficient where \( f < t \leq n-f \), similar to threshold signatures. The generation and verification of coin-shares are also non-interactive. We describe the security properties of the scheme and related algorithms in Appendix \ref{TCT}.
\end{enumerate}

\section{Design of pMVBA}
\subsection{pMVBA Overview}
In this section, we present a protocol for multi-valued Byzantine agreement capable of tolerating up to \( f < \frac{n}{3} \) Byzantine faults. The protocol features an expected communication bit complexity of \( O(ln^2 + \lambda n^2) \) and an expected constant asynchronous round count. Our modular implementation comprises five distinct sub-protocols: Committee Selection (CS), Prioritized Verifiable Consistent Broadcast (pVCBC), Propose-Recommend, Random Order/Permutation Generation, and Sequential-ABBA. The framework of the proposed protocol is illustrated in Figure \ref{fig:Framework}.

\subsection{Committee Selection}
In the proposed pMVBA protocol, the CS sub-protocol plays a critical role in optimizing efficiency by reducing the communication complexity inherent in Byzantine agreement protocols. Rather than requiring all $n$ parties to broadcast their proposals, our method strategically selects a subset of at least $f + 1$ parties, based on the party ID and the current instance of the protocol, to perform this task. This selection ensures the protocol’s progress while maintaining its security properties, with the output being a set of $f + 1$ parties designated as committee members for that specific instance.

To achieve this, the CS protocol utilizes a cryptographic coin-tossing scheme, inspired by Dumbo \cite{FASTERDUMBO}, to dynamically and randomly select $\kappa = f + 1$ parties for each protocol instance. This selection process guarantees that at least one honest party is included in the committee, thereby preserving the integrity of the protocol.

The dynamic nature of party selection mitigates several risks, including adversarial corruption, ensuring fair participation across all parties and reducing the likelihood of starvation and Denial-of-Service (DoS) attacks. The randomness introduced by the cryptographic coin-tossing scheme is crucial in preventing adversary influence and securing the CS process.

The CS protocol is illustrated in Algorithm \ref{algo:cs} and involves the following steps:

\begin{itemize}
    \item \textit{Generating Coin-Shares:} Upon invocation, each party generates a coin-share $\sigma_i$ and broadcasts it to all parties. The party then waits to receive $f + 1$ coin-shares. (lines 3-5)
    \item \textit{Verifying Coin-Shares:} Upon receiving a coin-share from another party $p_k$, the receiving party verifies the coin-share and adds it to the set $\Sigma$ until $f + 1$ valid shares are collected. (lines 8-10)
    \item \textit{Selecting Parties:} After collecting $f + 1$ valid coin-shares, the $CToss$ function is used to determine the selected subset of parties. (lines 6-7)
\end{itemize}

\begin{algorithm}[htbp]
\LinesNumbered
\DontPrintSemicolon
\SetAlgoNoEnd
\SetAlgoNoLine

\SetKwProg{LV}{Local variables initialization:}{}{}
\LV{}{
   $\Sigma \leftarrow \{\}$\;
   $CUR\_INSTANCE \leftarrow 0$\;
}

\SetKwProg{un}{upon}{ do}{}
\un{$CS \langle ID\rangle $ invocation}
{
  $CUR\_INSTANCE \leftarrow ID.INSTANCE$\;
  $id \leftarrow ID.id$ \;
  $\sigma _{id}$ $\leftarrow$ $CShare\langle id \rangle$ \;
  
  \textbf{send} $\langle SHARE, id, \sigma _{id}, CUR\_lINSTANCE \rangle$ to all parties\;
  \textbf{wait until} $|\Sigma| = f+1$\;
  \KwRet $CToss \langle id, \Sigma \rangle$\;
}

\un{receiving $\langle SHARE, id, \sigma _k, INSTANCE \rangle$ from a party $p_{k}$ for the first time}
{
 \uIf{ $CUR\_INSTANCE = INSTANCE$ and $CShareVerify\langle id, k, \sigma _k\rangle = true$ }{
    $\Sigma \leftarrow \Sigma \cup \{\sigma _k \}$}
}

\caption{Committee Selection for party $p_i$}
\label{algo:cs}
\end{algorithm}

\begin{figure*}[htbp]
     \centering
   \begin{subfigure} [b]{0.45\textwidth}
       \centering
       \includegraphics[width=0.9\textwidth]{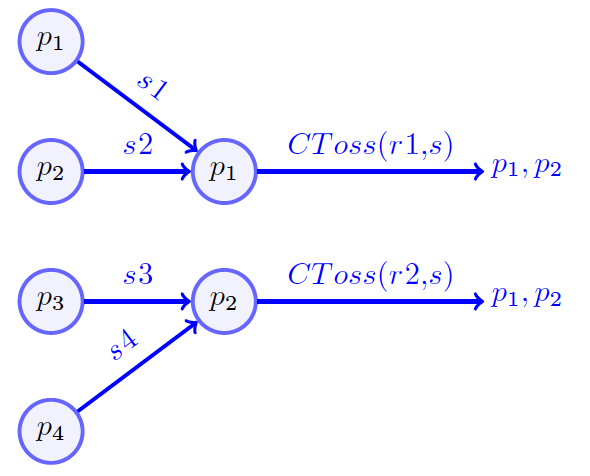}
         \caption{Committee-Selection protocol. Each party shares their $s_i = CShare(r_i)$, and a party requires $f + 1$ shares. The $CToss(r_1,s)$ function is used to select the parties, ensuring that each party gets the same set. This illustration shows the process for two parties, but it applies universally.}
         \label{fig:CS}
   \end{subfigure} \hspace{0.8em}
    \begin{subfigure}[b]{0.45\textwidth}
         \centering
         \includegraphics[width=0.9\textwidth]{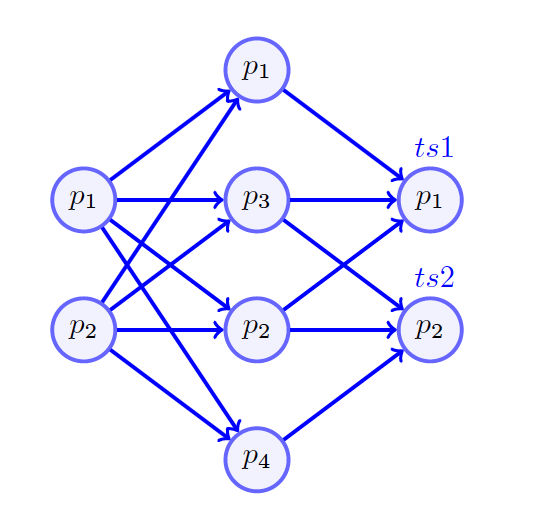}
         \caption{pVCBC protocol. Each selected party (here $p_1$ and $p_2$ ) proposes its requests. Each receiving party verifies the requests, adds its sign-share to the message, and returns it to the sender. A party waits for 2f+1 sign-shares (here it is 3 because f = 1). The 2f+1 sign shares are necessary to create a threshold-signature ($ts1$).  }
         \label{fig:pVCBC}
     \end{subfigure}
     \caption{Committee Selection and pVCBC protocol}
\end{figure*}

The CS protocol, as depicted in Figure \ref{fig:CS}, ensures a secure and efficient selection of broadcasting parties, which significantly reduces the overall communication complexity of the pMVBA protocol. Definitions of the cryptographic coin-tossing functions (e.g., $CShareVerify$) are provided in Appendix \ref{TCT}.

\subsection{pVCBC} Following the selection of committee members through the Committee Selection protocol, each selected member must provide a verifiable proof of their proposal to ensure that it has been broadcast to at least $f+1$ honest parties. Specifically, the input for this protocol includes the ID, requests, and the selected parties, while the output is a threshold-signature—a verifiable proof that the same requests have been sent to at least $f+1$  honest parties. This proof is crucial for maintaining the integrity and consistency of the protocol, as it verifies that the proposal has been correctly disseminated among the parties. Traditionally, the Verifiable Consistent Broadcast (VCBC) protocol is used to generate such proofs, enabling every party to produce a verifiable record of their broadcast proposals. The definition of verifiability and the VCBC protocol itself are detailed in Appendix \ref{VCBC}. However, since our protocol restricts broadcasting to only the selected committee members, we utilize a slightly modified version of the VCBC protocol, which we term pVCBC (Prioritized Verifiable Consistent Broadcast).

The pVCBC protocol aligns with the selective broadcasting approach from the Committee Selection process. This adaptation ensures that when a party receives verifiable proof from another, there's no need to verify its origin from a selected committee member, as the protocol inherently guarantees this. This streamlines verification and reduces unnecessary checks, preserving the efficiency introduced by Committee Selection. The pVCBC protocol's construction is detailed in Algorithm \ref{algo:per}, and an illustration is provided in Figure \ref{fig:pVCBC}, visually representing its interactions and verification steps as a critical component of the overall pMVBA protocol.

\paragraph*{Construction of the pVCBC} \label{pVCBC}
\begin{itemize}
    \item Upon invocation of the pVCBC protocol, a party sends (ID, requests, SelectedParties) to the parties. (lines 03).
    \item Upon receiving a message (ID, requests, SelectedParties) from a party $p_k$, a party checks whether the sender is a selected party. If the sender is a selected party, the party adds its signature share ($sign-share$) to the message, resulting in $\sigma$. The party then replies with $\sigma$ to the sender (lines 11-14).
    \item Upon receiving a signature share $\sigma_k$ from a party $p_k$, a selected party adds the signature share $\sigma_k$ to its set $\Sigma$ (lines 07 - 09).
    \item A selected party waits for $\langle n-f \rangle$ valid signature shares. Upon receiving $ n-f $ signature shares, the party combines them to generate a threshold signature $\rho$ (proof that the party has sent the same request to at least $ f+1 $ honest parties) and returns $\rho$ to the caller (line 05).
\end{itemize}

\begin{algorithm}[htbp]
\LinesNumbered
\DontPrintSemicolon
\SetAlgoNoEnd
\SetAlgoNoLine

\SetKwProg{un}{upon}{ do}{}
\un{pVCBC$\langle ID, requests, SelectedParties\rangle$ invocation} 
{
      $\Sigma \leftarrow \{\}$ \;
     \textbf{send} $\langle ID, requests \rangle$ to all parties\;
     \textbf{wait until} $|\Sigma| = n-f$\;
     \KwRet $\rho \leftarrow CombineShare_{id}\langle requests, \Sigma \rangle$\;
  
}
\;

\SetKwProg{un}{upon}{ do}{}
\un{receiving$\langle requests, \sigma_{k}\rangle $ from the party $p_{k}$ for the first time} {
\uIf{$VerifyShare_{k}\langle requests, (k, \sigma_k)  \rangle$ }{
    $\Sigma \leftarrow {\Sigma  \cup \{\sigma_k\}}$
 }
}
\;
\un{receiving $\langle ID, requests\rangle $ from the party $p_{j}$ for the first time} {
\uIf{$p_j \in SelectedParties$ }{
   $\sigma_{id} \leftarrow SigShare_{id} \langle sk_{id}, requests \rangle$\;
   $reply \langle requests,  \sigma_{id} \rangle$\;
 }
}
\;
\;

\SetKwProg{LV}{Local variables initialization:}{}{}
\LV{}{
   $\Sigma \leftarrow \{\}$\;
}

\SetKwProg{un}{upon}{ do}{}
\un{$RandomOrder \langle ID, SelectedParties\rangle $ invocation}
{
  $\sigma _i$ $\leftarrow$ $CShare\langle r_{id} \rangle$ \;
  $id = ID.id$\;
  $\textbf{send} \langle SHARE, id, \sigma _{id}, ID.INSTANCE \rangle$ to all parties\;
  \textbf{wait until} $|\Sigma| = 2f+1$\;
  
  \KwRet filterSelectedParties( $CToss \langle r_{id}, \Sigma \rangle, SelectedParties$)\;
}
\;
\SetKwProg{un}{upon receiving}{ do}{}
\un{$ \langle SHARE, id, \sigma _k, INSTANCE \rangle$ from a party $p_{k}$ for the first time}
{
 \uIf{$CShareVerify\langle r_{k}, k, \sigma _k\rangle = true$ }{
    $\Sigma \leftarrow { \Sigma \cup  \{\sigma _k\}}$}
}

\;
\;
\caption{pVCBC and RandomOrder protocols for party  $p_i$}
\label{algo:per}
\end{algorithm}

\subsection{Propose-Recommend}

In the pMVBA protocol, the Propose-Recommend step replaces the traditional commit step found in the classical MVBA protocol. This shift is central to enhancing the protocol's efficiency by significantly reducing complexity. While the classic commit step involves collecting \(n - f\) verifiable proofs and creating an array of length \(n\), our recommend step simplifies the process by broadcasting verifiable proofs directly. The propose step remains largely similar to the one in the classic protocol, but with a key difference: only the selected parties, as determined by the CS process, are responsible for proposing their requests. The input for this step includes a proposal and a threshold-signature, but it does not produce any direct output. 

The recommend step then follows, where the input consists of the proposal and the threshold-signature. The output of this step is a list of threshold-signatures (collected from the \(n-f\) recommendation message), ensuring that at least one verifiable proof reaches the majority of parties. This process is critical for moving the protocol towards an agreement.

The recommend step involves the following key actions:

\begin{itemize}
    \item \textit{Initialization:} A party initializes an empty set $\Sigma$. (line 2)
    \item \textit{Broadcast Recommendation:} A party creates a recommendation message and broadcasts it to all parties. (line 6)
    \item \textit{Recommendation Collection:} Upon receiving a recommendation, a party verifies the threshold-signature and instance. If valid, the party adds the proposal and $\sigma$ to its set $\Sigma$. (lines 11-14)
    \item A party waits for $n - f$ recommendation messages before proceeding. (line 6-7)
\end{itemize}

The detailed SCR (Send, Collect, Recommend) protocol is shown in Algorithm \ref{alg:RC}, which formalizes the steps involved in the recommend process.

\begin{algorithm}[htbp]
\SetAlgoLined
\DontPrintSemicolon
\textbf{Local variables initialization:}\;
\LinesNumbered
$\Sigma \leftarrow \{\}$\;
$CUR\_INSTANCE \leftarrow 1$\;
$rCount \leftarrow 0$\;

\SetKwProg{un}{upon}{ do}{}
\un{$SCR\langle ID, ID_k, Proposal, \rho \rangle$ invocation}
{
  \textbf{send} $\langle RECOMMENDATION, ID_k, Proposal, \rho \rangle$ to all parties\;
  $CUR\_INSTANCE \leftarrow ID.INSTANCE$\;
  $rCount \leftarrow 0$\;
  \textbf{wait until} $rCount = n-f$\;
  \KwRet $\Sigma$\;
}

\un{receiving $\langle RECOMMENDATION, ID_j, Proposal, \rho \rangle$ from party $p_{j}$ for the first time} {
\uIf{  $ID.INSTANCE = CUR\_INSTANCE$ $\wedge$ $Verify_{id}\langle Proposal, \rho \rangle = true$ $\wedge ID_j.id \notin \Sigma$}{
    $\Sigma \leftarrow {\Sigma \cup  \{ID_j.id: \langle Proposal, \rho \rangle\} }$\;
    
    $rCount \leftarrow rCount + 1$\;
 }
}

\caption{Recommend Protocol for all parties}\label{alg:RC}
\end{algorithm}

\subsection{Random Order}
Following the Propose-Recommend step, the next crucial phase in the pMVBA protocol is to determine the order in which the selected parties will proceed. The purpose of this Random Order step is to generate a random permutation of the selected parties, ensuring that the process remains fair and secure. This step is particularly important for preventing adversaries from manipulating the message delivery order, which could potentially increase the number of asynchronous rounds required to reach an agreement.

The input for this step includes the party ID and the current protocol instance, and the output is a permutation of the selected parties. The random order is generated after the distribution and receipt of the threshold number of recommendation messages, ensuring the process is not vulnerable to adversarial interference. The Random Order protocol involves the following steps :

\begin{itemize}
    \item \textit{Coin-Share Generation:} Upon invocation of the RandomOrder protocol, a party generates a coin-share $\sigma_i$ for the current instance and broadcasts it to every party. The party then waits for $2f+1$ coin-shares. (lines 19-24)
    \item \textit{Coin-Share Verification:} When a party receives a coin-share from another party $p_k$ for the first time, it verifies the coin-share, ensuring it is from $p_k$, and accumulates the coin-share in the set $\Sigma$. The party continues to respond to coin-shares until it has received $2f+1$ valid shares. (lines 27-28)
    \item \textit{Permutation Generation:} Once $2f+1$ valid coin-shares have been received, the party uses its $CToss$ function and the collected coin-shares to generate a permutation of the $n$ parties and filtered the selected parties and return. (line 24 of Algorithm \ref{algo:per}).
\end{itemize}

This process is formalized in the pseudocode for the random-order generation protocol, as shown in Algorithm \ref{algo:per}. By generating a random order of the selected parties, this step ensures that the protocol remains robust against adversarial manipulation and maintains the efficiency of the agreement process.

\subsection{Sequential-ABBA}

The final phase in the pMVBA protocol is the Sequential-ABBA protocol, which is responsible for reaching an agreement on one of the proposals submitted by the selected parties. Building on the random order generated in the previous step, the Sequential-ABBA protocol runs an agreement loop that systematically evaluates each selected party's proposal. The input to this protocol consists of the permutation list generated in the Random Order step and the list of recommendations obtained from the Propose-Recommend step. The output is the final proposal agreed upon by the parties.

In Sequential-ABBA, we employ an asynchronous binary Byzantine agreement protocol that is biased towards $1$, allowing for efficient agreement on proposals. The protocol iterates through the permutation list, running the agreement process for each selected party until an agreement is achieved on a single proposal. This approach ensures that the final decision is reached in a manner that is both robust and fair, in line with the protocol's overall design. The detailed construction of the Sequential-ABBA protocol is provided below, where the specific steps and mechanisms are described. This protocol serves as the culminating step in the pMVBA process, ensuring that all parties converge on a single, agreed-upon proposal, thus completing the agreement procedure.

\paragraph*{Construction of the Sequential-ABBA} \label{Sequential-ABBA}
Since the Sequential-ABBA protocol guarantees that the parties will eventually reach an agreement on one of the selected party's proposals, the protocol runs a loop for the selected parties until it agrees on a valid proposal. The pseudocode of the Sequential-ABBA protocol is given in Algorithm \ref{algo:Sequential-ABBA}, and a step-by-step description is provided below:

\begin{itemize}
    \item The Sequential-ABBA protocol takes two arguments: (i) PermutationList, the permutation of the selected parties, and (ii) RecommendationList, the list of verifiable proofs and proposals the party has received (line 1).
    \item Upon invocation of the protocol, a party declares two variables: (i) $index$, an index number to access the selected parties one by one from the array PermutationList, and (ii) $bit$, initially set to zero to indicate that the parties have not reached any agreement (lines 2-3).
    \item \textbf{While loop:} Parties are chosen one after another according to the permutation PermutationList of $\{1, \ldots, f+1\}$. Let $index$ denote the index of the party selected in the current loop (the selected party $p_{index}$ is called the candidate). Each party $p_i$ follows these steps for the candidate $p_{index}$ (lines 4-25):
    \begin{itemize}
        \item Broadcasts a $VOTE$ message to all parties containing $u_{index} = 1$ if party $p_i$ has received $p_{index}$’s proposal and verifiable proof (including the proposal in the VOTE or from RecommendationList), and $u_{index} = 0$ otherwise. (lines 12-18)
        \item Waits for $\langle n-f \rangle$ $VOTE$ messages (line 19) but does not count votes indicating $u_{index} = 1$ unless a valid proposal from the party $p_{id}$ has been received—either directly or included in the $VOTE$ message (lines 21 - 22).
        \item Runs a binary asynchronous Byzantine agreement biased towards $1$ (see Algorithm \ref{algo:ABBA}) to determine whether $p_{index}$ has properly broadcast a valid proposal. Vote $1$ if $p_i$ has received a valid proposal from $p_{index}$ through RecommendationList and add the protocol message that completes the verifiable broadcast of $p_{index}$’s proposal to validate this vote. Otherwise, if $p_i$ has received $n-f$ $VOTE$ messages containing $u_{index} = 0$, then vote $0$; no additional information is needed. If the agreement loop decides $1$, exit the loop. (lines 21-25)
    \end{itemize}
    \item Upon reaching an agreement, if the proposal is empty, then use the threshold signature $\rho$ to get the proposal from the selected party and return the proposal. (lines 26-28)
    
    \item Upon receiving $\langle ID, VOTE, party, u_{index}, m \rangle$, a party checks whether $u_{index'} = 1$. If it receives $u_{index'} = 1$, then it assigns $m_{index} = m$. (lines 29-33)
\end{itemize}

\begin{algorithm}[htbp]
\DontPrintSemicolon
\SetAlgoNoEnd
\SetAlgoNoLine
\SetKwProg{un}{upon receiving}{ do}{}

\SetKwProg{un}{upon}{ do}{}
\un{Sequential-ABBA$\langle ID, PermutationList, RecommendationList \rangle$ invocation}
{
  $index \leftarrow 0$\;
  $bit \leftarrow 0$\;

 \While{bit=0}{
    $\Sigma \leftarrow \{\}$\;
    
    $index \leftarrow index + 1$\;

    $proposal \leftarrow \perp$\;
    $\rho \leftarrow \perp$\;
    
    $id \leftarrow PermutationList [index] $\;
    $u_{index} \leftarrow 0$\;
    $m_{index} \leftarrow <proposal, \rho>$\;
    

     \uIf{$id \in RecommendationList$}{

        $u_{index} \leftarrow 1$ \;
        $<proposal, \rho> \leftarrow RecommendationList[id]$\;
        $m_{index} \leftarrow <proposal, \rho>$\;
        \textbf{send} $\langle ID, VOTE, party_{id}, u_{index}, m_{index} \rangle$ to all parties\; 
    }\uElse{
        \textbf{send} $\langle ID, VOTE, party_{id}, u_{index}, \perp \rangle$ to all parties\;
       
    }

    \textbf{ wait until} $|\Sigma|$ $= n-f$ \;
     \;
    
    \uIf{$ u_{index}=1$}{
     $v \leftarrow \langle 1, <proposal, \rho> \rangle$\;
    }\uElse{
     $v \leftarrow \langle 0, \perp \rangle$\;
    }
    $\langle bit, \rho \rangle \leftarrow ABBA\langle v \rangle$ biased towards $1$\;
   }

   \uIf{$proposal = \perp$} {
     use $\rho$ to complete the verifiable authenticated broadcast and deliver the proposal. See Appendix \ref{VCBC}
   }
 return proposal \;
}

\un{receiving $\langle ID, VOTE, p_k, u_{index'}, m \rangle$ for the first time from party $p_k$} {
      \uIf{$u_{index'} = 1$ }{
    $u_{index} \leftarrow 1$ \;
    $m_{index} \leftarrow m$\;
 }

 $\Sigma \leftarrow \Sigma \cup 1 $\;
}

\caption{Protocol for party  $p_i$ }
\label{algo:Sequential-ABBA}
\end{algorithm}


\subsection{Integration of Subprotocols}
The pMVBA protocol achieves an agreement in a Byzantine environment through a sequence of interconnected steps. The process begins with the Committee Selection (CS) protocol, where each party dynamically and randomly selects committee members (Algorithm \ref{alg:PMVBA}, line 10). Following this, the selected parties broadcast their proposals using the Prioritized Verifiable Consistent Broadcast (pVCBC) protocol, ensuring that their broadcasts are verifiable (Algorithm \ref{alg:PMVBA}, line 13). After the proposals are broadcast, the parties proceed with the Propose-Recommend phase, where verifiable proofs are shared, and recommendations are collected (Algorithm \ref{alg:PMVBA}, lines 23-27). Upon receiving the necessary recommendations, a random order of the selected committee members is generated using a cryptographic coin-tossing scheme (Algorithm \ref{alg:PMVBA}, line 18). Finally, the Sequential-ABBA protocol is executed, where the agreement process is run in the determined order until an agreement is reached on a proposal (Algorithm \ref{alg:PMVBA}, lines 20-21). Each of these steps ensures the protocol operates efficiently and securely, even in the presence of Byzantine faults.

\begin{algorithm}[htbp]
\LinesNumbered
\DontPrintSemicolon
\SetAlgoNoEnd
\SetAlgoNoLine

\SetKwProg{LV}{Local variables initialization:}{}{}
\LV{}{
$ INSTANCE \leftarrow 1$\;
$ SelectedParties \leftarrow \bot$\;
$\rho (threshold-signature)\leftarrow \bot$\;
$Proposal \leftarrow \bot$\;
$RecommendationList \leftarrow \bot$\;
$PermutationList \leftarrow \bot$ \;
}

\SetKwProg{un}{upon}{ do}{}
\SetKwProg{Fn}{function}{}{}
\While{true}
{
   $ID \leftarrow {\langle INSTANCE, id \rangle}$\;
   $SelectedParties \leftarrow CS \langle ID, INSTANCE \rangle$\;

    \uIf{$p_i \in SelectedParties$}{
       $Proposal \leftarrow requests$\;
       $ \rho  \leftarrow pVCBC \langle ID, Proposal\rangle$\;
       $\textbf{send}\langle PROPOSE, ID, \rho, Proposal \rangle$ to all parties\;
    }
    \textbf{wait until}  $|RecommendationList| \geq 1$ \; \;

    \textbf{Permutation:}\;
     $PermutationList \leftarrow RandomOrder\langle ID, SelectedParties \rangle$\;

    \textbf{Sequential-ABBA:}\;
    $Proposal \leftarrow Sequential$-$ABBA \langle PermutationList, RecommendationList \rangle$\;
    $decide \langle Proposal \rangle$\;
    $INSTANCE \leftarrow INSTANCE+1$\;
}

\un{receiving $\langle \textit{PROPOSE}, ID_k, \rho, Proposal\rangle$ for the first time}{
         $  RecommendationList \leftarrow SCR\langle ID, ID_k, Proposal, \rho \rangle$\;
}

\un{receiving $\langle \textit{RECOMMEND}, ID_k, Proposal,  \rho \rangle$ for the first time}{
         \uIf{no $RECOMMENDATION$ is sent yet }{
             $RecommendationList \leftarrow SCR\langle ID, ID_k,Proposal, \rho \rangle$\;
         }
    }

\caption{Prioritized-MVBA: Protocol for party $p_i$}
\label{alg:PMVBA}
\end{algorithm}


\section{Evaluation and Analysis of the Proposed Protocol}
\label{sec:eval}
This section evaluates the proposed protocol based on its correctness and efficiency. Correctness ensures that the protocol maintains predefined security properties, making it resilient against adversarial attacks. Efficiency evaluates the protocol's performance in terms of resource utilization and execution time. In this Section, we conduct extensive evaluation and analysis to answer:

\noindent
\textbullet~~What are the formal proofs supporting the correctness of the pMVBA protocol? ($\S $ \ref{subsection:security})

\noindent
\textbullet~~How does the pMVBA protocol perform in terms of message complexity, communication complexity, and running time? ($\S $ \ref{subsection:efficiency})

\noindent
\textbullet~~What are the behaviors and outcomes of the protocol under different network conditions and adversarial strategies? ($\S $ \ref{subsection: case_study})

\noindent
\textbullet~~How does the performance of the pMVBA protocol compare to other MVBA protocols in terms of key metrics? ($\S $ \ref{subsection: performance_metric})

\noindent
\textbullet~~How does pMVBA's performance compare to atomic broadcast protocols in communication complexity? ($\S $ \ref{ComparisonABC})


\subsection{Proof of the Prioritized-MVBA Protocol}
\label{subsection:security}

The correctness of the proposed protocol is critical to ensure that it adheres to the standard MVBA outcomes. Our goal is to prove that reducing the number of broadcasts does not compromise the protocol's reliability. The output of a party's request depends on the output of an ABBA instance within our protocol's agreement loop, typically resulting in a value of \textit{1}. For the ABBA protocol to output \textit{1}, at least one honest party must input \textit{1}. We substantiate our protocol's integrity by demonstrating that at least $2f+1$ parties receive a verifiable proof of a party's proposal, thus ensuring the necessary input for the ABBA instance.

\begin{lemma}
\label{lemma1}
At least one party's proposal reaches $2f+1$ parties.
\end{lemma}
\begin{proof} Since both the selected and non-selected parties can be non-responsive (Byzantine or system failure), and the adversary can coordinate with the Byzantine parties to schedule the message delivery in a way that slows down the agreement process. To prove the lemma, we parameterized the number of responsive parties among the selected parties and the total number of responsive parties among both selected and non-selected parties. Below is the parameterized statement. 
\begin{itemize}
    
    \item Among the $f+1$ selected parties,  $t$ parties are responsive where $1\leq t \leq  f+1$ , and among total $n$ parties, $m$ number of parties are responsive, where $ 2f+1 \leq m \leq n $.
    We provide three scenarios below, which are special cases of the above statement. We utilize the special cases to prove the parameterized statement. The special cases are three edge cases where the maximum number of selected and the non-selected parties can be responsive and non-responsive (case 1 and 2)
    \begin{enumerate}
         \item Among $ f+1$ selected parties, $f$  parties are non-responsive. (The maximum number of selected parties can be non-responsive)
         \item Selected $ f+1 $ parties are responsive, but other $f$ non-selected parties are non-responsive. (The maximum number of non-selected parties can be non-responsive)
         \item Every party is responsive, including the selected $f+1$ parties. (The maximum number of responsive parties)
    \end{enumerate}
    We will first prove that the above three scenarios are a special case of the parameterized statement. These three cases are the lower and upper bound for $t$ and $m$.

    \begin{enumerate}
         \item For case 1, $t=1$ (selected responsive parties) and $m = t+2f = 2f+1$ (since $f$ selected parties are non-responsive, the total number of responsive parties is $2f+1$ that includes one responsive selected party). So it adheres requirement of the parameterized equation.
         \item For case 2, $t=f+1$ and $m = t+f = f+1+f (f+1 selected responsive party and f non-seleted responsive party)= 2f+1$. So it also adheres requirement of the parameterized equation.
         \item For case 3, $t=f+1$ and $m=t+2f=3f+1=n$. Therefore, it also adheres requirement of the parameterized equation.
    \end{enumerate}
\end{itemize}

We prove that in every scenario, at least one party's proposal reaches $\langle 2f+1 \rangle$ parties. Since we have proved that cases (1), (2), and (3) are the special cases of a parameterized statement, we utilize these cases to prove that even in the scenario of parameterized statement, at least one proposal reaches $2f+1$ parties.

\begin{enumerate}

    \item The proof proceeds by contradiction. Assume that among the $f+1$ selected parties, $f$ number of parties are nonresponsive, meaning only one party completes the $pVCBC$ protocol and broadcasts the verifiable proof. Under the assumption that no proposal reaches more than $2f$ parties, the total number of recommendations would be:

    \hspace{12.0em} $(1)(2f)(\frac{2f+1}{1})$ 

    On average, a selected party’s proposal is recommended by $\frac{2f+1}{1}$ parties. Simplifying this:

    \begin{equation}
       \begin{split}
         (1)(2f)(\frac{2f+1}{1} ) \\ 
         = (2f)(2f + 1) \\
       \end{split}
    \end{equation}

    Thus, the total number of recommendations is $(2f)(2f + 1)$, which is less than the required $(2f + 1)(2f + 1)$ recommendations.

    This leads to a contradiction, as the protocol necessitates a sufficient number of recommendation messages to guarantee progress. Given that honest parties are obligated to provide the necessary number of recommendations, it follows that the proposal must have been received by a minimum of $2f + 1$ parties. Consequently, the assumption that no proposal reaches more than $2f$ parties is incorrect, and at least one selected party's proposal reaches $2f + 1$ parties, thereby enabling the protocol to function as intended.
    
    
    \item We will prove by contradiction. Assume that no proposal reaches more than $2f$ parties.

    Given the assumption that a minimum of $2f+1$ parties are responsive and that $f+1$ of these are chosen to propose, it follows that every party must receive a proposal during the $propose$ step. Furthermore, no party commences the Sequential-ABBA protocol until it has received a minimum of $2f+1$ recommendations and has secured the random order necessary to advance with Sequential-ABBA. This indicates that a total of $(2f+1)(2f+1)$ recommendation messages must be exchanged.

    Assuming that no proposal is recommended to more than $2f$ parties and that $f+1$ parties submit their requests, the total number of recommendations would be:

\hspace{12.0em} $(f+1)(2f)(\frac{2f+1}{f+1})$ 

 On average, a proposal from a selected party would receive recommendations from $\frac{2f+1}{f+1}$ parties. Simplifying the expression:

\begin{equation}
\begin{split}
 (f+1)(2f)\frac{2f+1}{f+1}  \\ 
 =(f+1)(2f)(2 - \frac{1}{f+1})\\
 = 4f(f+1) - \frac{2f(f+1)}{f+1} \\
 =4f^2 + 4f -2f \\
 =2f(2f + 1) \\
\end{split}
\end{equation}

Therefore, the total number of recommendations is $(2f)(2f+1)$, which is lesser to the necessary $(2f+1)(2f+1)$ recommendation messages.

This results in a contradiction, as the protocol necessitates the transmission of sufficient recommendation messages for progress to occur. Honest parties are required to ensure the delivery of the necessary number of recommendations; thus, at least one proposal must be received by $2f+1$ parties. Therefore, our assumption that no proposal reaches more than $2f$ parties is incorrect. Consequently, at least one party's proposal successfully reaches $2f+1$ parties, thereby ensuring the protocol can advance as necessary.
    

    \item The proof proceeds by contradiction. Assume that no proposal reaches more than $2f$ parties.

    Since we assume that all parties are responsive, each party receives a proposal during the $propose$ step. Additionally, no party engages in the Sequential-ABBA protocol until it has received a minimum of $2f+1$ recommendations and acquired the random order required to execute the Sequential-ABBA protocol. Thus, the protocol requires a total of $(3f+1)(2f+1)$ recommendation messages.

If no proposal is recommended to more than $2f$ parties and $f+1$ parties propose their requests, then the total number of recommendations must not exceed $(f+1)(2f)$. On average, a proposal from a selected party can receive recommendations from $\frac{3f+1}{f+1}$ parties. Therefore, the total number of recommendations does not exceed:

    \begin{equation}
        \begin{split}
           (f+1)(2f)(\frac{3f+1}{f+1} )
        \end{split}
    \end{equation}
 
   Simplifying:

   \begin{equation}
     \begin{split}
        (f+1)(2f)\frac{3f+3-2}{f+1}\\
        = (f+1)(2f)(3- \frac{2}{f+1}\\
        = 6f(f+1)-4f\\
        = 6f^2+2f\\
        = 2f(3f+1)\\
     \end{split}
   \end{equation}

   Therefore, the total number of recommendations can be expressed as $(2f)(3f+1)$. The necessary quantity of recommendations is $(3f+1)(2f+1)$.

   While a proposal can be recommended by more than one party, if it is assumed that all recommending parties recommend the same set of $2f$ parties, the condition of receiving $2f+1$ recommendations remains unfulfilled. Consequently, for the protocol to advance, it is essential that the honest parties broadcast a sufficient number of recommendation messages, which the adversary will ultimately deliver.

   Therefore, a minimum of one proposal must be communicated to $2f+1$ parties. This contradicts the initial assumption; therefore, at least one party's proposal has to reach $2f+1$ parties, thereby ensuring the protocol's advancement.

\end{enumerate}
We prove that in every scenario, at least one party's proposal has been shared to $\langle 2f+1 \rangle$ parties. Having shown that cases (1), (2), and (3) are particular instances of the general parameterized scenario, the proof of these specific cases allow us to extend the result to the parameterized scenario.

In the three scenarios analyzed, it is evident that with $m$ responsive parties, the total number of recommendations are $m  (2f+1)$. If $t$ selected parties are responsive, the total number of recommendations can be represented as $t ( 2f) ( \frac{m}{t})$. In this context, $(t) (2f)$ indicates that the proposals from each of the $t$ selected parties can reach a maximum of $2f$ parties. Meanwhile, $\frac{m}{t}$ denotes the average number of parties that each selected party's proposal can reach.

Therefore, we establish the subsequent relationship:

\begin{equation}
    \begin{split}
        (t)(2f)\frac{m}{t} < (m)(2f+1)\\  
    \end{split}
\end{equation}

Simplifying further:
\begin{equation}
    \begin{split}
        (m)(2f) < (m)(2f+1)\\
    \end{split}
\end{equation}

This inequality confirms that the number of recommendations from the selected parties is less than the required total number of recommendations, proving that at least one party’s proposal must reach $2f+1$ parties, ensuring protocol progress.

\end{proof}

This lemma ensures that if $2f+1$ number of parties receive verifiable proof, at least one honest party will input \textit{1} to the ABBA instance, ensuring the protocol reaches an agreement on \textit{1}.

\begin{lemma}
\label{lemma2}
Without any permutation, the adversary can cause at most $f+1$ iterations of the agreement loop in the Sequential-ABBA protocol.
\end{lemma}
\begin{proof}
Since only $f+1$ parties are selected, the number of iterations is bounded above by $f+1$. By Lemma~\ref{lemma1}, at least one verifiable proof is delivered to $f+1$ honest parties, which guarantees that the protocol reaches agreement on at least one party’s proposal. We therefore analyze the worst case in which the protocol requires exactly 
$f+1$ iterations. If the ABBA execution order of the selected parties is known in advance, an adaptive adversary can schedule message deliveries so that only the proposal of the last party in that order attains recommendations from $2f+1$ number of parties. Consequently, for all earlier selected parties, less than the threshold number of parties cast $vote=1$ and the instance decides $0$; for the final selected party, at least $f+1$ honest parties cast $vote = 1$ and the instance decides $1$.
.
\end{proof}
\begin{lemma}
\label{lemma3}
Let $\overline{A} \subseteq \{1, 2, ..., f+1\}$ be the set of selected parties for which at least $f+1$ honest parties receive the verifiable proof, and let $\Pi$ be a random permutation of the $f+1$ selected parties. Then, except with negligible probability:
\begin{itemize}
    \item For every party $p \in \overline{A}$, the ABBA protocol on $ID|p$ will decide $1$.
    \item $|\overline{A}| \geq 1$.
    \item There exists a constant $\beta > 1$ such that for all $t \geq 1$, $Pr[\Pi[1] \notin \overline{A} \wedge \Pi[2] \notin \overline{A} \wedge ... \wedge \Pi[t] \notin \overline{A}] \leq \beta^{-t}$.
\end{itemize}
\end{lemma}
\begin{proof}
The binary agreement protocol biased towards $1$ decides $0$ if the honest parties input $0$. From the Sequential-ABBA protocol, each party must receive $\langle n-f \rangle$ $vote=0$ messages in order to input $0$ to an ABBA instance and to decide $0$. However, from Lemma \ref{lemma1} at least $\langle f+1 \rangle$ honest parties receive at least one verifiable proof from a selected party $p$, then it is not the case for the party $p\in \overline{A}$. This proves the first claim. 

To prove the second claim, we refer to the two special cases of Lemma \ref{lemma1}. In case $1$, one selected party is responsive and every party receives the verifiable proof of that selected party; therefore, $|\overline{A}| = 1$. For case 2, every party receives a recommendation for each selected party. Therefore, each selected party's verifiable proof reaches $2f+1$ parties, thus $|\overline{A}| = f+1$. Since Lemma \ref{lemma1} proves that at least one party's proposal reaches $2f+1$ parties, it generalizes the claim $|\overline{A}| \geq 1$.

The third claim follows now because $|\overline{A}|$ is at least a constant fraction of $f+1$ and thus, there is a constant $\beta > 1 $ such that $Pr[\Pi(i)) \notin \overline{A}] \leq \frac{1}{\beta}$ for all $1 \leq i \leq t$. Since the probability of the $t$ first elements of $\Pi$ jointly satisfying the condition is no larger than for $t$ independently and uniformly chosen values, we obtain $Pr[\Pi[1] \notin \overline{A} \wedge \Pi[2] \notin \overline{A} \wedge .... \wedge \Pi[t] \notin \overline{A}] \leq \beta^{-t}$ \footnotemark.

\footnotetext{To prove the part of the Lemma \ref{lemma3}, we have used the same technique of Lemma 9 from \cite{SECURE03}. We also refer interested readers to Lemma 10 of the same paper for the pseudorandom generation.}
\end{proof}

\begin{theorem}
Given a protocol for biased binary Byzantine agreement and a protocol for verifiable consistent broadcast, the Prioritized-MVBA protocol provides multi-valued validated Byzantine agreement for $n > 3f$ and invokes a constant expected number of binary Byzantine agreement protocols.
\end{theorem}

\renewcommand\qedsymbol{$\blacksquare$}
\begin{proof}
\textbf{Agreement:} If an honest party outputs a value $v$, then every honest party outputs the value $v$. From Lemma \ref{lemma1}, at least one proposal reaches $2f+1$ parties, and the ABBA protocol reaches an agreement on $1$ if $f+1$ honest parties input $1$. The agreement property of the ABBA protocol ensures all honest parties output $1$ for that ABBA instance and receive the same value $v$, which satisfies $threshold-validate\langle v, \rho \rangle = true$.

\textbf{Liveness:} If honest parties participate and deliver messages, all honest parties decide. From Lemma \ref{lemma1}, at least one proposal reaches $2f+1$ parties, ensuring the ABBA protocol decides $1$. Lemma \ref{lemma3} confirms the protocol reaches an agreement after a constant expected number of iterations.

\textbf{External-validity:} If an honest party terminates and decides on a value $v$, then $externally-valid\langle v, \rho \rangle = true$. The validity of the ABBA protocol ensures at least one honest party inputs $1$ to the ABBA instance, meaning the honest party received a valid $threshold-signature$ $\rho$ for $v$.

\textbf{Integrity:} If honest parties decide on a value $v$, then $v$ was proposed by a party. The ABBA protocol returns $1$ if at least one honest party inputs $1$, which requires a valid $threshold-signature$ $\rho$. Honest parties reply with $sign-shares$ only if a value $v$ is proposed, ensuring $v$ was proposed by a party.
\end{proof}

\subsection{Efficiency Analysis}
\label{subsection:efficiency}

The efficiency of a Byzantine Agreement (BA) protocol depends on message complexity, communication complexity, and running time. We analyze the proposed protocol's efficiency by examining its sub-components: the pVCBC sub-protocol, committee-selection and permutation generation, propose-recommend steps, and the Sequential-ABBA sub-protocol.

\paragraph*{Running Time:} Each sub-protocol and step, except for Sequential-ABBA, has a constant running time. The running time of the proposed protocol is dominated by the Sequential-ABBA sub-protocol. The Sequential-ABBA sub-protocol runs the ABBA protocol biased towards $1$ in an expected constant number of times (Algorithm \ref{algo:Sequential-ABBA} takes a permutation list of the selected parties as an argument and iterates over that list). Since we take a permutation of the selected parties (see Algorithm \ref{algo:per} line 19-24), the probability of choosing a selected party is an expected constant number. The ABBA protocol also has an expected constant number of asynchronous rounds. Therefore, the running time of the protocol is an expected constant number.

\paragraph*{Message Complexity:} In all sub-protocols and steps, except for pVCBC and propose steps, each party communicates with all other parties. Every party transmit $O(1)$ information to all other parties (See line 8 of Algorithm \ref{algo:cs}, line 6 of Algorithm \ref{alg:RC}, lines 14 of Algorithm \ref{alg:PMVBA}, lines 3, 14, and 22 of Algorithm \ref{algo:per}, lines 16 and 18 of Algorithm \ref{algo:Sequential-ABBA}. Each of the broadcasts sends $O(1)$ information). Since $n$ parties send $O(1)$ information to the $n$ parties, the message complexity is $O(n^2)$. The expected message complexity of the Sequential-ABBA protocol is also $O(n^2)$.

\paragraph*{Communication Complexity:} The communication complexity of each sub-protocol and step is $O(n^2(l + \lambda))$, where $l$ is the bit length of input values and $\lambda$ is the bit length of the security parameter. To calculate the communication complexity we use the same approach as message complexity. We observe that in no step a party transmit $O(n)$ information. Thus, the communication complexity is same as message complexity only includes the bit length of the input values and the bit length of the security parameters. The expected communication complexity of the Sequential-ABBA protocol is also $O(n^2(l + \lambda))$.

\paragraph*{Worst-Case Analysis:} Table \ref{table:Table_IV} presents the analysis of worst-case rounds for each base protocol. The results show that the worst-case rounds are consistent in terms of \(f\) across all the base protocols.

\subsection{Case Study}\label{subsection: case_study}

The key contribution of this paper is the reduction in the number of requesting parties and the subsequent utilization of fewer broadcasts to eliminate the expensive sub-components of the classic MVBA protocol. The main challenge is to provide sufficient information to the parties to maintain the protocol's progress while removing the costly elements. It is crucial to ensure that at least one selected party's proposal reaches at least $2f+1$ parties. This is achieved through the pVCBC step and the propose-recommend step. To illustrate the protocol's effectiveness, we present a case study in two parts. Section \ref{CaseStudy1} demonstrates how the protocol achieves the desired properties with the minimum number of nodes a system can have. To further clarify that the protocol can maintain these properties with a larger number of nodes, Section \ref{CaseStudy2} provides a case study using charts for a system with three times the number of faulty nodes.


\subsubsection{Message Flow in Different Stages}\label{CaseStudy1}
In an asynchronous network, the presence of a faulty node and the adversary's ability to manipulate message delivery can delay the agreement process. This section shows how a faulty party can affect message delivery patterns and how the adversary can delay specific messages to prevent a node's proposal from reaching the majority. Despite these challenges, the goal is to ensure that nodes can still reach an agreement on a party's proposal.

Figure \ref{fig:CaseStudy1fig1} assumes no faulty nodes, with the adversary delivering messages uniformly, allowing the protocol to reach an agreement on the first try. Figure \ref{fig:CaseStudy1fig2} considers a scenario where one node gets a threshold-signature ($ts$) early, or the adversary prioritizes one node's $ts$ delivery over others, preventing one selected node's proposal from reaching the majority. Figure \ref{fig:CaseStudy1fig3} assumes the selected node $p_2$ is either faulty or completely isolated by the adversary, and Figure \ref{fig:CaseStudy1fig4} assumes a non-selected node is faulty or isolated, positively impacting the $ts$ delivery of the selected nodes.

\begin{figure*}[htbp]
     \centering
     \begin{subfigure}[b]{0.45\textwidth}
         \centering
         \includegraphics[width=\textwidth]{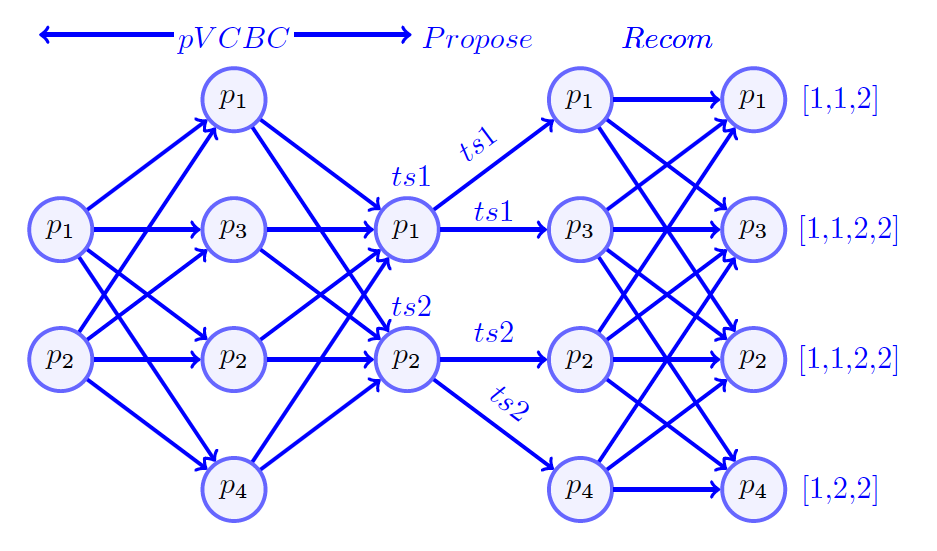}
         \caption{Every selected node is non-Byzantine and the adversary delivers the threshold-signature (ts1, ts2) uniformly (both $p_1$ and $p_2$'s threshold-signature are received by the same number of parties). Both $p_1$'s and $p_2$'s proposals are received by every party.}
         \label{fig:CaseStudy1fig1}
     \end{subfigure}
     \hfill
     \begin{subfigure}[b]{0.45\textwidth}
         \centering
         \includegraphics[width=\textwidth]{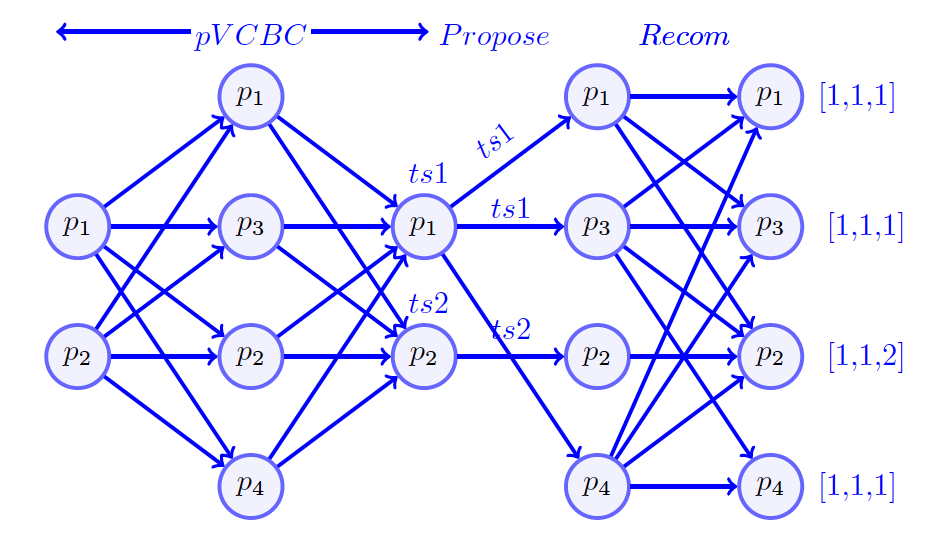}
         \caption{Every selected node is non-Byzantine and the adversary delivers the messages non-uniformly (Node $p_1$'s threshold-signature is received by three parties but node $p_2$'s threshold-signature is received by one party). $p_1$'s proposal is received by every party, but node $p_2$'s proposal is received by only one party.}
         \label{fig:CaseStudy1fig2}
     \end{subfigure}
     \caption{Message flow without the presence of faulty nodes.}
\end{figure*}   

\begin{figure*}[htbp]
     \centering
     \begin{subfigure}[b]{0.45\textwidth}
         \centering
         \includegraphics[width=\textwidth]{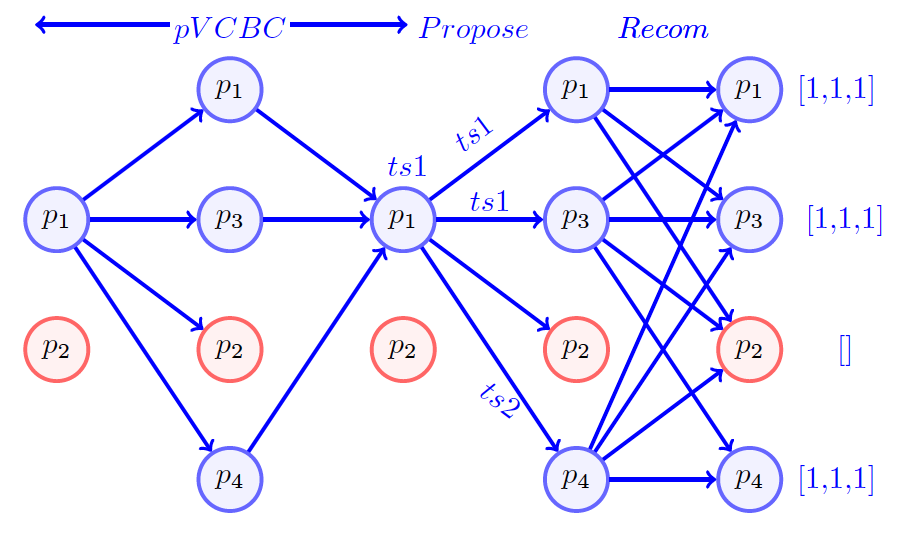}
         \caption{One selected node is non-Byzantine ($p_1$), and the adversary delivers the messages from this node to every other node. Consequently, only node $p_1$ completes the pVCBC protocol and proposes the threshold-signature. Therefore, every party has the proposal except $p_2$ (either Byzantine or non-responding). The red node represents either the node is Byzantine or non-responding or the network is not delivering the messages from that node.}
         \label{fig:CaseStudy1fig3}
     \end{subfigure}
     \hfill
     \begin{subfigure}[b]{0.45\textwidth}
         \centering
         \includegraphics[width=\textwidth]{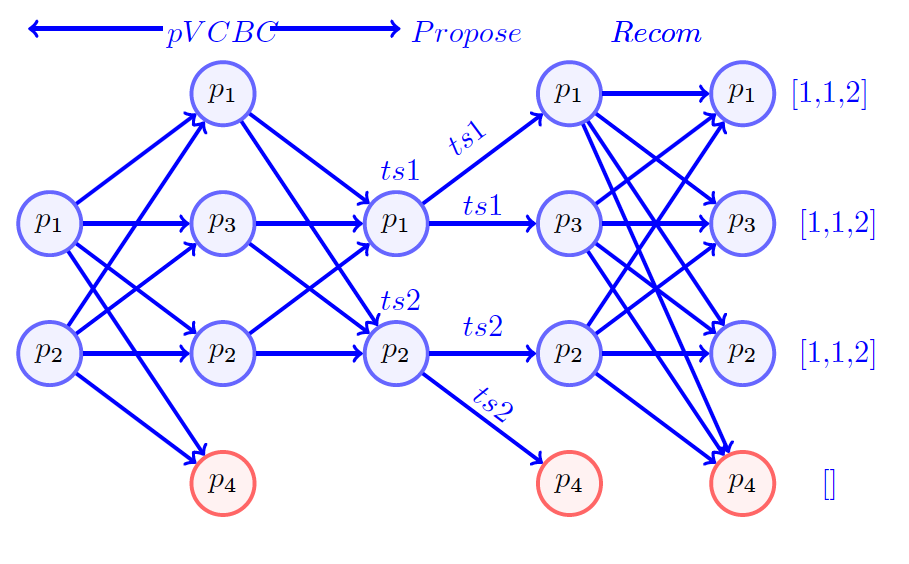}
         \caption{Every selected node is non-Byzantine, and the adversary delivers the messages uniformly (both $p_1$ and $p_2$'s threshold-signature are received by the same number of parties). Both $p_1$'s and $p_2$'s proposals are received by every party. The non-selected node $p_4$ is not responding, which does not affect the overall criteria. The red node represents either the node is Byzantine or non-responding or the network is not delivering the messages from that node.}
         \label{fig:CaseStudy1fig4}
     \end{subfigure}
     \caption{Message flow in the presence of faulty nodes.}
\end{figure*}

In conclusion, the above figures collectively demonstrate that regardless of the message delivery pattern or the presence of faulty nodes, the protocol consistently achieves agreement, thereby proving its robustness and effectiveness.

\subsubsection{Message Distribution for More Than One Faulty Node}\label{CaseStudy2}

For this study, we assume a total of $n=3f+1=10$ parties, with $f=3$ being faulty. This configuration allows us to explore various message distribution patterns. We require at least one party's proposal to reach $2f+1=7$ parties, with a total of $f+1=4$ selected parties. We examine the impact of different numbers of faulty parties and message distribution patterns. In each figure, the top box indicates selected parties and the number of parties that can recommend the proposal to others (e.g., Party $p_1$ can receive 3 recommendations for the first selected party ($R_1$), 3 for the second selected party ($R_2$), and 2 for the third selected party ($R_3$), totaling 7). The bottom box counts the number of recommendations received for each selected party (e.g., $p_1$ receives a total of 7 $R_1$ recommendations).

Figure \ref{fig:CSFig1} assumes no faulty nodes, with the adversary delivering messages uniformly, ensuring that every selected node's proposal reaches the threshold number of parties. Figure \ref{fig:CSFig2} assumes three non-selected non-responding nodes, which also allows every selected node's proposal to reach the threshold. Figure \ref{fig:CSFig3} assumes one honest selected node or the adversary delivering messages from that node, ensuring only the honest selected node's proposal reaches the threshold. Figure \ref{fig:CSFig4} and Figure \ref{fig:CSFig5} show scenarios with two and three honest selected nodes, respectively, where their proposals reach the threshold. Figure \ref{fig:CSFig6} explores non-uniform message distribution, which may prevent some selected nodes' proposals from reaching the threshold, but only when all selected nodes are active.

In conclusion, the above scenarios illustrate that the protocol can maintain its effectiveness and achieve agreement even with varying numbers of faulty nodes and different message distribution patterns.

\begin{figure*}[htbp]
    \centering
    \hspace{1ex}
    \begin{minipage}[t]{0.3\linewidth}
        \centering
        \includegraphics[width=\linewidth]{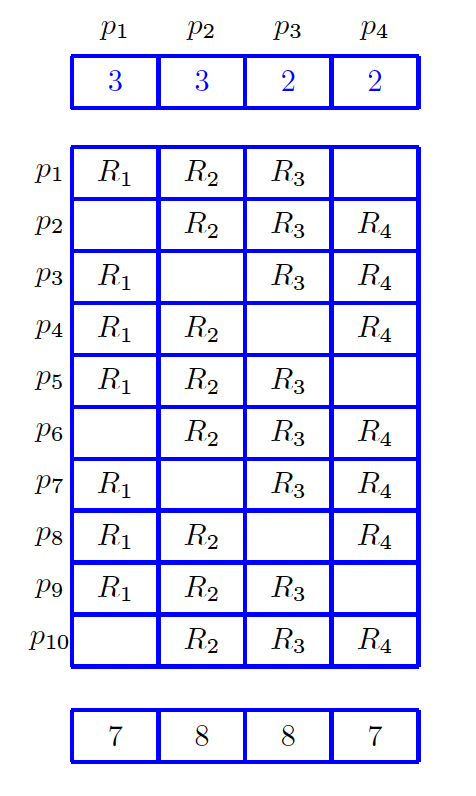}
        \caption{All selected parties complete their pVCBC and broadcast their proposals. All parties are non-Byzantine; therefore, all of them recommend their received proposal. Every party's proposal reaches at least 7 parties.}
        \label{fig:CSFig1}
    \end{minipage}
    \hfill
    \begin{minipage}[t]{0.3\linewidth}
        \centering
        \includegraphics[width=\linewidth]{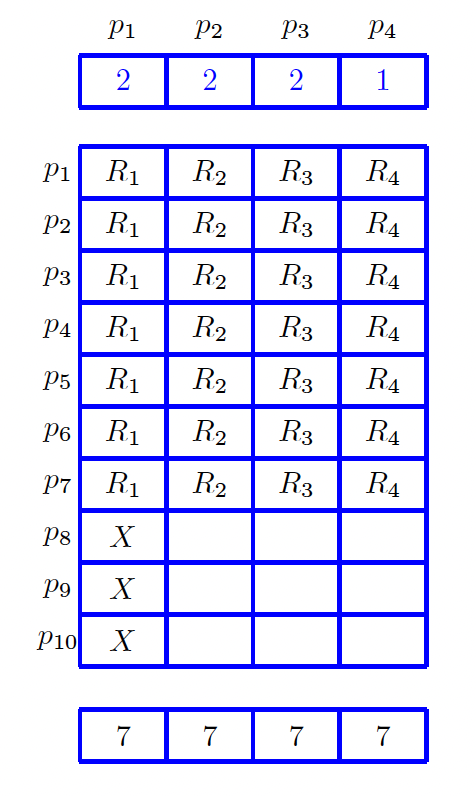}
        \caption{All selected parties complete their pVCBC and broadcast their proposals. Three non-selected faulty parties cannot recommend; consequently, a party receives recommendations for all selected parties. Therefore, all selected parties' proposals reach at least 7 parties.}
        \label{fig:CSFig2}
    \end{minipage}
    \hfill
    \begin{minipage}[t]{0.3\linewidth}
        \centering
        \includegraphics[width=\linewidth]{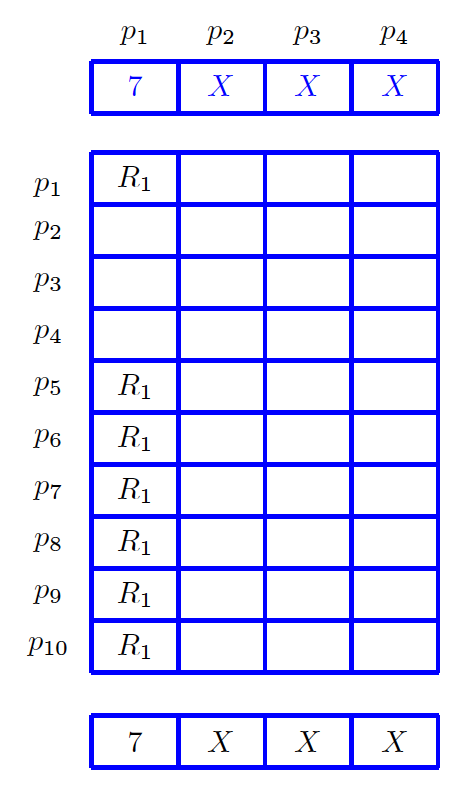}
        \caption{All selected parties are Byzantine or non-responding except $P_1$. The non-Byzantine party completes the pVCBC and proposes its requests. Since there is only one proposal, every party receives the same proposal, which is received by every non-Byzantine party. Therefore, the proposal reaches 7 parties.}
        \label{fig:CSFig3}
    \end{minipage}
\end{figure*}

\begin{figure*}[htbp]
    \centering
    \hspace{1ex}
    \begin{minipage}[t]{0.3\linewidth}
        \centering
        \includegraphics[width=\linewidth]{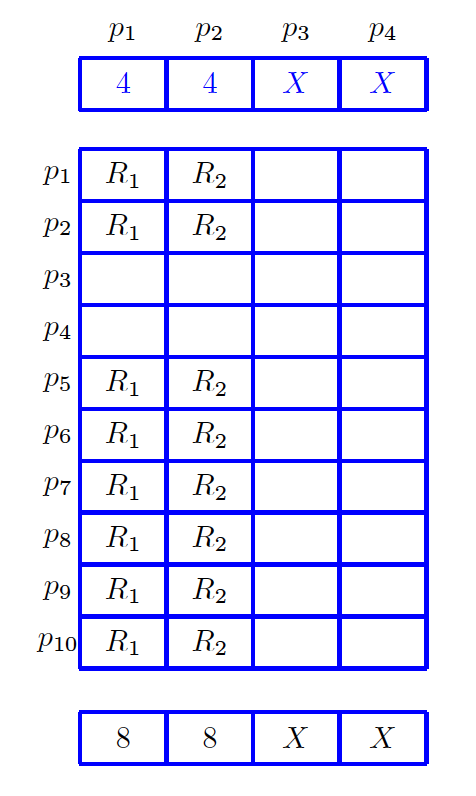}
        \caption{Two selected parties are non-Byzantine and complete the pVCBC protocol, proposing their proposals, which are received by every non-Byzantine party. Each selected party's proposal reaches more than 7 parties.}
        \label{fig:CSFig4}
    \end{minipage}
    \hfill
    \begin{minipage}[t]{0.3\linewidth}
        \centering
        \includegraphics[width=\linewidth]{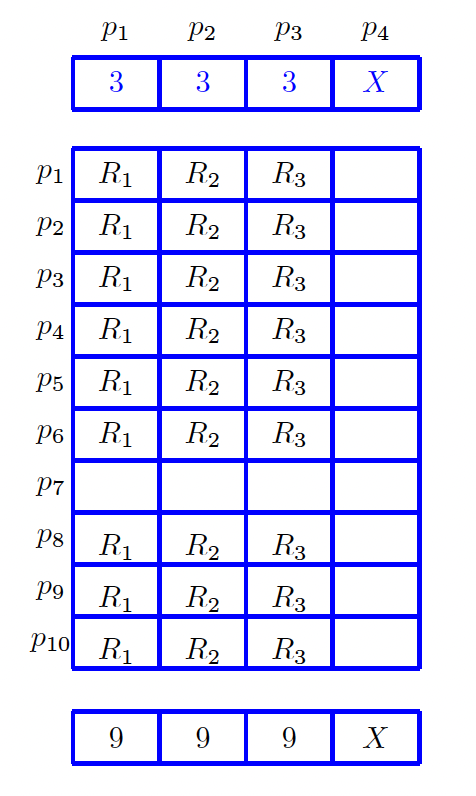}
        \caption{Three selected parties are non-Byzantine and complete the pVCBC protocol, proposing their proposals, which are received by every non-Byzantine party. Each selected party's proposal reaches more than 7 parties.}
        \label{fig:CSFig5}
    \end{minipage}
    \hfill
    \begin{minipage}[t]{0.3\linewidth}
        \centering
        \includegraphics[width=\linewidth]{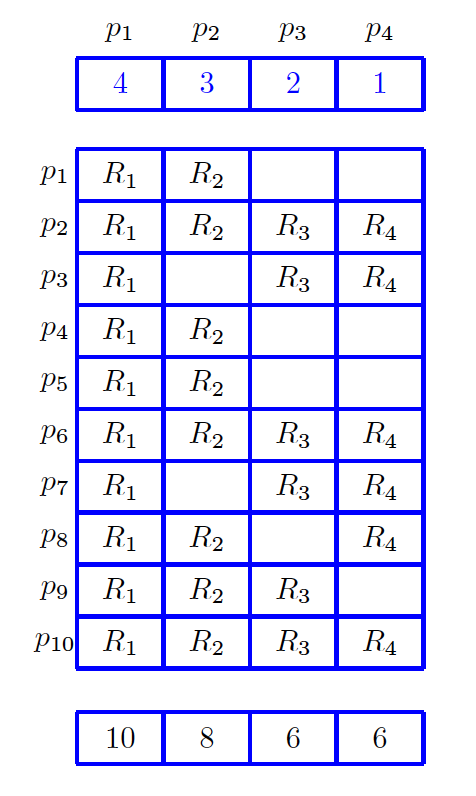}
        \caption{This study tests whether non-uniform message distribution affects recommendation reception. Even with non-uniform distribution, at least half of the selected parties' proposals reach 7 parties.}
        \label{fig:CSFig6}
    \end{minipage}
\end{figure*}
\newpage
\subsection{Analysis of Communication Complexity and Resilience} \label{subsection: performance_metric}
This section evaluates our protocol's performance by comparing its communication complexity with classic MVBA protocol and assessing its resilience, termination, and safety properties against committee-based protocols.

\subsubsection{Comparison of Asynchronous Rounds in MVBA Protocols}
We compare the asynchronous rounds required to reach agreement in our pMVBA protocol with those in classic and traditional MVBA protocols, considering both best-case and worst-case scenarios. To determine the worst-case cost, we first calculate the best-case scenario, identify the concrete rounds, and then estimate the remaining rounds as average ABBA rounds, multiplying by \(f+1\). As shown in Table \ref{table:Table_IV}, Dumbo-MVBA requires the most rounds in the best-case scenario. In the worst-case, while all protocols require rounds proportional to \(f\), VABA struggles to guarantee agreement on a complete output after \(f+1\) instances, requiring 13 rounds per instance—twice as many as other protocols. Among existing MVBA protocols, pMVBA offers superior performance in communication complexity and asynchronous rounds. Although Speeding Dumbo achieves lower asynchronous rounds, it is an ACS protocol and cannot avoid the \(O(\lambda n^3 \log n)\) communication complexity.

\begin{table*}[htbp]
\centering
\caption{Asynchronous rounds of MVBA protocols}
\begin{tabular}{||c | c |c ||} 
 \hline
 Protocols &  Rounds (Best case) & Rounds (Worst case) \\ [0.5ex] 
 \hline\hline
 Cachin et al. \cite{SECURE03} & $\boldsymbol{12}$ & $\geq 5+7(f+1)$  \\ 
 [0.5ex] 
 \hline
 VABA \cite{BYZ17} & $13$  & $\geq 13(f+1)$  \\ 
 \hline
 Dumbo-MVBA \cite{BYZ20} & $18$ &  $\geq 11+7(f+1)$  \\ 
 \hline
 \textbf{Our work} & $13$ & $\boldsymbol{6+ 7(f+1)}$ \\
 \hline
\end{tabular}

\label{table:Table_IV}
\end{table*}

\subsubsection{Comparison of communication complexity with MVBA protocols.}
Our work focuses on the first polynomial-time asynchronous Byzantine agreement protocol and removes its expensive sub-component. We compare our work with the classic and traditional MVBA protocols, as highlighted in Table \ref{table:Table_I}. Our techniques differ from recent improvements in MVBA protocols, so we compare our protocol with the classic protocol and the following improvement. Our protocol differs from the Cachin-MVBA \cite{SECURE03} in communication complexity. VABA \cite{BYZ17} is a view-based protocol where each view is a complete protocol execution. VABA does not guarantee agreement in every instance, whereas our protocol ensures parties reach an agreement on a valid proposal in each instance.

We also present a comparison between our proposed protocol and the Dumbo-MVBA \cite{BYZ20} protocols. The Dumbo-MVBA protocol is intended for large inputs. It uses the erasure code technique to remove the expensive term. The erasure code first disperses the message and then recovers the message at the end of the agreement, which is expensive in terms of latency and computation and is not worth it for small-size inputs. Our protocol achieves the desired communication complexity (removal of the $O(n^3)$ term) without the utilization of information dispersal techniques (erasure code) and in one instantiation of the protocol. We can use the information dispersal technique like Dumbo-MVBA in our protocol to reduce the communication complexity to $O(ln + \lambda n^2)$, but our protocol does not depend on the erasure code to remove the $O(n^3)$ term. We use a committee to remove the $O(n^3)$ term it takes only one message and the cryptographic computation. The cryptographic computation is similar to other parts of the protocol. 

\begin{table}[htbp]
\begin{center}

\caption{Comparison for performance metrics of MVBA protocols (E. Code stands for Erasure Code and Inst stands for Instance)}

\begin{tabular}{||c |c |c |c |c | c |c ||} 
 \hline
 Protocols & Comm. (bit) & Word & Time & Msg. & E. Code. & Inst\\ [0.4ex] 
 \hline\hline
 Cachin \cite{SECURE03} & $O(ln^2 + \lambda n^2 + n^3)$ & $O(n^3)$ & $O(1)$ & $O(n^2)$ & No & 1\\ 
 \hline
 VABA \cite{BYZ17} & $O(ln^2 + \lambda n^2)$ & $O(n^2)$ & $O(1)$ & $O(n^2)$ & No & $\geq$ 1\\
\hline
 Dumbo-MVBA \cite{BYZ20} & $\boldsymbol{O(ln + \lambda n^2)}$ & $O(n^2)$ & $O(1)$ & $O(n^2)$ & Yes & 1\\
 \hline
 \hline
  \textbf{Our work} & $O(ln^2 + \lambda n^2)$ & $\boldsymbol{O(n^2)}$ & $\boldsymbol{O(1)}$ & $\boldsymbol{O(n^2)}$ & $\boldsymbol{No}$ & $\boldsymbol{1}$\\ [1ex] 
 \hline
\end{tabular}
\label{table:Table_I}
\end{center}
\end{table}

\subsubsection{Comparison of Resilience, Termination, and Safety with Committee-Based Protocols}
We compare our work with notable committee-based protocols, specifically focusing on resilience, termination, and safety properties. Table \ref{table:Table_II} highlights these comparisons. COINcidence \cite{BYZ19} assumes a trusted setup and violates optimal resilience. It also does not guarantee termination and safety with high probability (whp). Algorand \cite{BYZ21} assumes an untrusted setup, with resilience dependent on network conditions, and does not guarantee termination whp. The Dumbo \cite{FASTERDUMBO} protocol uses a committee-based approach, but its committee-election protocol does not guarantee the selection of an honest party, thus failing to ensure agreement or termination. Our protocol achieves optimal resilience and guarantees both termination and safety, as our committee-election process ensures the selection of at least one honest party. This guarantees that the protocol can make progress and reach an agreement despite adversarial conditions.

\begin{table}[htbp]
    
\caption{Comparison for performance metrics of the committee based protocols}
\begin{center}
\begin{tabular}{||c |c |c |c ||} 
 \hline
 Protocols & n$>$ & Termination & Safety\\ [0.4ex] 
 \hline\hline
 COINcidence \cite{BYZ19} & 4.5f & whp & whp\\ 
 \hline
 Algorand \cite{BYZ21} & * & whp & w.p. 1\\ 
 \hline
 Dumbo1 \cite{FASTERDUMBO} & 3f & whp & w.p. 1\\ 
 \hline
 Dumbo2 \cite{FASTERDUMBO} & 3f & whp & w.p. 1\\ 
 \hline
 \textbf{Our work} & \textbf{3f} & \textbf{w.p. 1} & \textbf{w.p. 1}\\ [1.0ex]
 \hline
\end{tabular}
\label{table:Table_II}
\end{center}
\end{table}

In conclusion, we have compared our protocol with both the classical protocol and the committee-based protocol. Though our protocol differs from the atomic broadcast protocol in a number of proposals, we provide a comparison of our protocol with the atomic broadcast protocol in Appendix \ref{ComparisonABC}.

\section{Implementation and Evaluations}
We implemented and evaluated pMVBA protocol using a self-developed gRPC based protocol simulator. Along the way, we also implemented and made systematic comparisons with several typical MVBAs, including VABA \cite{BYZ17} and CKPS01 \cite{SECURE03}.
\subsection{Performance of pMVBA Protocol }
We extensively evaluated and compared our proposed pMVBA protocol with VABA and Classic MVBA (CKSP01) protocol.
\paragraph{Implementation details.}All protocols are written as a single-process \textit{Python 3} program.  The \textit{p2p} channels among nodes are established using
gRPC \textit{insecure} channel. The \textit{Python multiprocessing} library handles the concurrent tasks.  Different from \cite{HoneyBadgerBFT} that optimistically skips
the verification for threshold signatures, our tests verify all signatures (simulate time to calculate threshold-signatures), which better reflects the actual performance in the worst cases, e.g., with corruptions. 

\paragraph{Test environment.}The experiments are conducted using gRPC instances with a similar
configuration. All tests scale up to $10$ instances. A transaction in our tests is a string of 1 KB, which approximates the size of 4 basic Bitcoin transactions with one input and two outputs. The batch size (also called the system load sometimes) represents the number of transactions proposed by all nodes in a one-shot agreement, and will vary from 1 to $5\times10^2$.

\paragraph{Result highlights.}The key information from our experiments is that the pMVBA significantly reduces the latency of Classic MVBA (CKSP01) with the increase of system sizes and attains a throughput $72\%$ higher than CKSP01, almost regardless of the system size, cf. Table \ref{table:Table_HIGHLIGHT};

\begin{table}[h!]
   
\caption{  Improvements of latency and throughput}
\begin{center}
\begin{tabular}{|c | c | c| c |c | c| c |} 
 \hline
 & \multicolumn{3}{ |c| }{Peak throughput (tx/s)} & \multicolumn{3}{ |c |}{Basic Latency (sec)}\\
 \hline
 Scale& CKSP01 & \multicolumn{2}{|c|}{pMVBA} & CKSP01 & \multicolumn{2}{|c|}{pMVBA}\\
\hline
n=4   & 15.1  & 9 & $\uparrow$ 68$\%$  & 6.72 & 5.71 & $\downarrow$ 18$\%$ \\
\hline
n=7 &  23.8 & 12 & $\uparrow$ 98$\%$ &  17 & 11 & $\downarrow$ 55$\%$\\
\hline
n=10 &  37.87  & 22 & $\uparrow$ 72$\%$ & 33 & 23 & $\downarrow$ 43$\%$\\
\hline

\end{tabular}
\label{table:Table_HIGHLIGHT}
\end{center}
\end{table}

\begin{figure}[h!]
     \centering
     \begin{subfigure}[b]{0.75\textwidth}
         \centering
         \includegraphics[width=\textwidth]{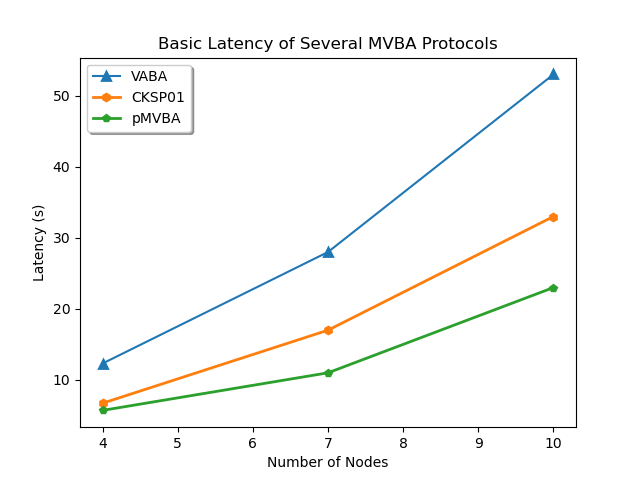}
        
     \end{subfigure}
     
      \caption{Latency (s) vs Number of Nodes}
      \label{fig:basiclatency}

\end{figure}   
\paragraph{\textbf{Basic latency.}}
 We set the input batch size as $4$ transactions
(i.e., nearly zero) to test the basic latency of pMVBA (and VABA and CKSP01) and show the results in Fig. \ref{fig:basiclatency}. It is easy to see that the improvement is consistent with the number of nodes. Though the basic
latency of VABA, CKSP01, and pMVBA is increasing with larger system scales; the improvement made by pMVBA is always significant.

\paragraph{\textbf{Throughput at varying scales.}}
We then evaluate the throughput of pMVBA at different system scales with varying batch sizes. Along the way, we also compare pMVBA to VABA and CKSP01 and visualize the comparison results in Fig. \ref{fig:throughput}.

As illustrated in the first three subgraphs in Fig. \ref{fig:throughput}, pMVBA attains multi-fold improvements relative to CKSP01, disregarding the system scale and batch size. First, pMVBA
directly inherits the advantages of using the 1-step VCBC protocol, after that, using fewer VCBC instances leads to fewer cryptographic computations and message delay time, which helps it finally achieve a throughput twice that of CKSP01 in the larger batch sizes and larger number of nodes. Moreover, carefully analyzing the Fig. \ref{fig:throughput}, we notice that pMVBA noticeably outperforms both VABA and CKSP01 with larger batch sizes and with the increasing number of nodes. The performance improvement is due to the smaller number of computations and messages needed for the pMVBA protocol.

The peak throughputs of all MVBA protocols are comprehensively shown in the last subgraph of Fig.  \ref{fig:throughput}. Although the throughput of all protocols decreases with scaling up, pMVBA always maintains the highest peak throughput among them.
\begin{figure*}[h!]
     \centering
     \begin{subfigure}[b]{0.48\textwidth}
         \centering
         \includegraphics[width=\textwidth]{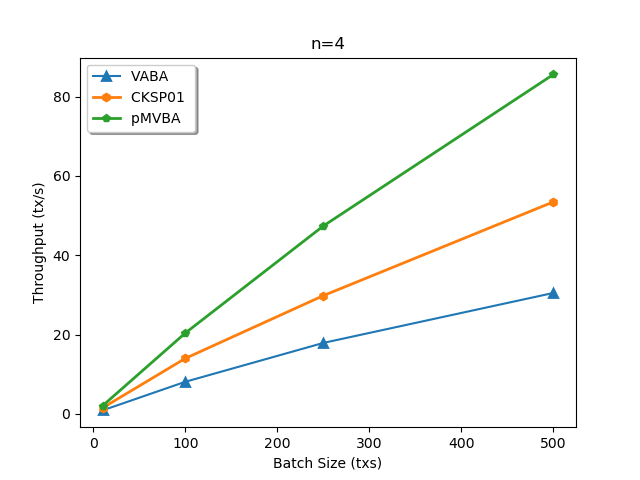}
         \caption{}
         \label{fig:throughput1}
     \end{subfigure}
     \hfill
     \begin{subfigure}[b]{0.48\textwidth}
         \centering
         \includegraphics[width=\textwidth]{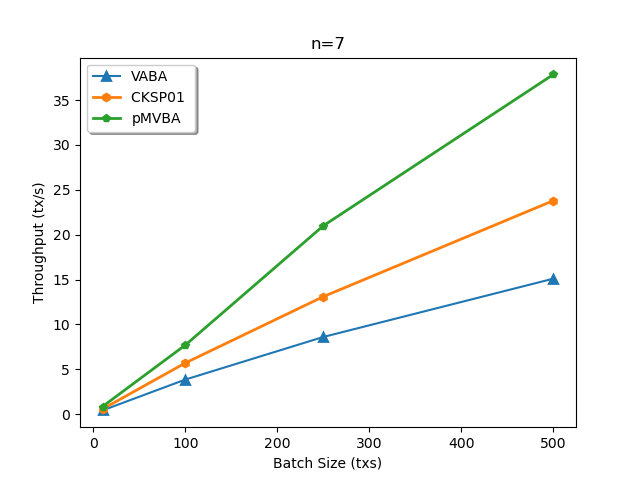}
         \caption{}
         \label{fig:throughput2}
     \end{subfigure}
     \vfill
     \begin{subfigure}[b]{0.48\textwidth}
         \centering
         \includegraphics[width=\textwidth]{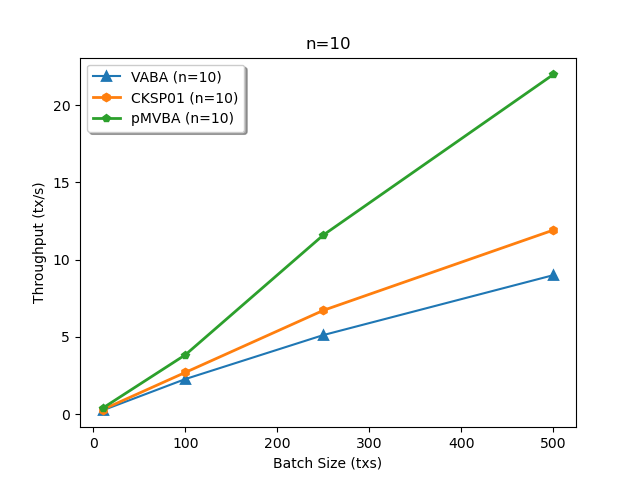}
         \caption{}
         \label{fig:throughput3}
     \end{subfigure}
     \hfill
     \begin{subfigure}[b]{0.48\textwidth}
         \centering
         \includegraphics[width=0.89\textwidth, height = 0.727\textwidth]{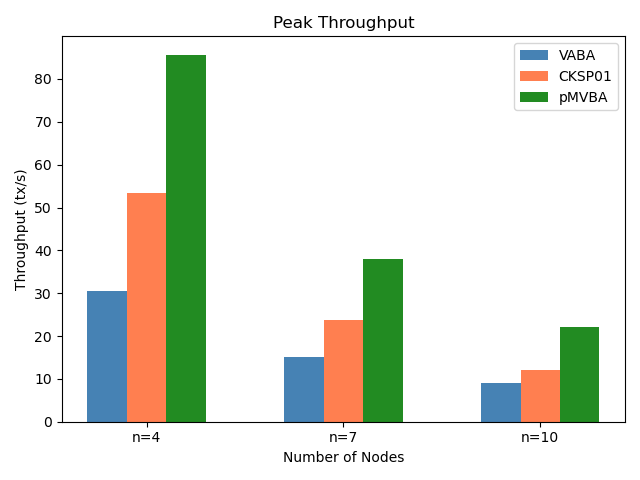}
         \caption{}
         \label{fig:throughput4}
     \end{subfigure}
     \caption{Throughput of Several MVBA Protocols}
     \label{fig:throughput}

\end{figure*}

\paragraph{Throughput-latency trade-off.} Fig. \ref{fig:LaThgh} shows the throughput latency trade-offs in pMVBA, CKSP01, and VABA. All the trends are roughly straight-line, indicating that all of them eventually become network-bound. pMVBA not only achieves higher peak throughput, but also has a latency that is always smaller than that of VABA and CKSP01 at the same throughput. Moreover,
the latency gap between pMVBA and VABA is greatly
narrowed at some small scales (VABA) and large scales (pMVBA). This evidence indicates that our multi-fold improvements greatly expand the applicability of the asynchronous BFT protocol.

\begin{figure*}[h!]
     \centering
     \begin{subfigure}[b]{0.48\textwidth}
         \centering
         \includegraphics[width=\textwidth]{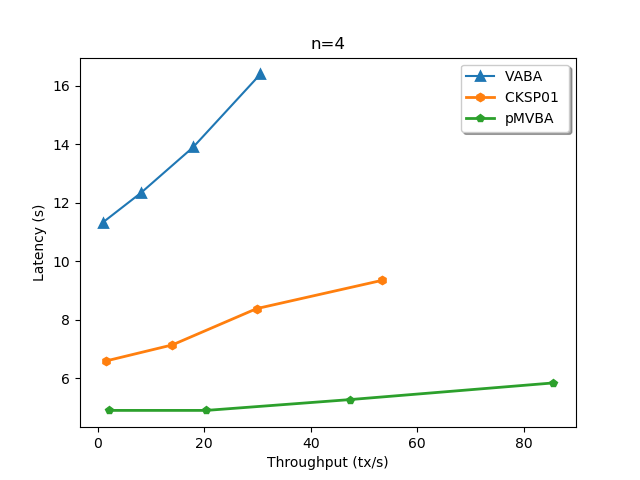}
         \caption{}
         \label{fig:throughput1}
     \end{subfigure}
     \hfill
     \begin{subfigure}[b]{0.48\textwidth}
         \centering
         \includegraphics[width=\textwidth]{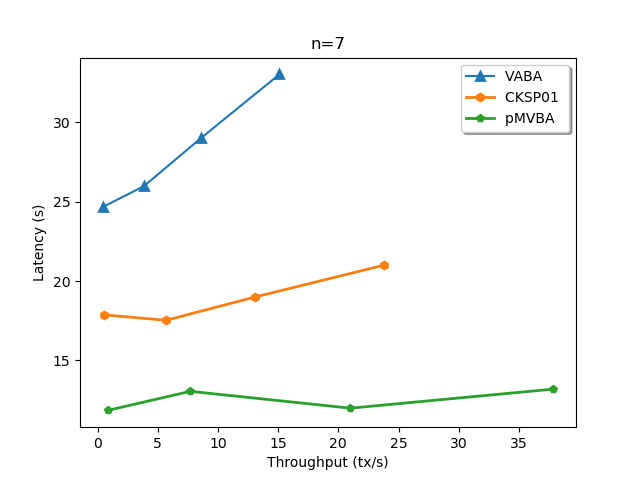}
         \caption{}
         \label{fig:LaThgh2}
     \end{subfigure}
     \vfill
     \begin{subfigure}[b]{0.48\textwidth}
         \centering
         \includegraphics[width=\textwidth]{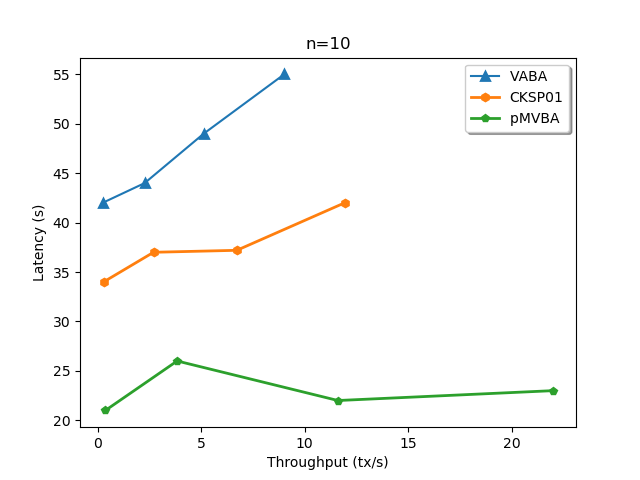}
         \caption{}
         \label{fig:LaThgh3}
         
     \end{subfigure}
     \hfill
     \begin{subfigure}[b]{0.48\textwidth}
         \centering
         \includegraphics[width=\textwidth]{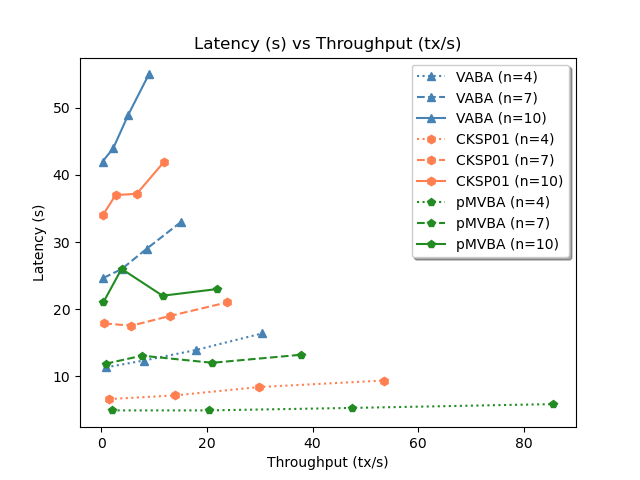}
         \caption{}
         \label{fig:LaThgh4}
        
     \end{subfigure}
    
     \caption{Latency (ts) vs throughput (tx/s) of the MVBA Protocols.} 
     \label{fig:LaThgh}
     
\end{figure*}

\section{Discussions and Conclusion}
In this paper, we introduced pMVBA, a novel MVBA protocol designed to address the high communication complexity inherent in the classic MVBA protocol. By leveraging a committee-based approach combined with the ABBA protocol, pMVBA significantly reduces communication overhead while maintaining optimal resilience and correctness. Our protocol achieves several key improvements over existing MVBA protocols. By dynamically selecting a subset of parties (f + 1) to broadcast proposals, we ensure that at least one honest party is always included, thereby enhancing the protocol's resilience. The integration of the pVCBC protocol allows for efficient proposal broadcasting and verifiable proof generation. The recommend step replaces the commit step used in the classic MVBA protocol, reducing the need for extensive communication bit exchanges and achieving agreement with fewer communication rounds. Theoretical analysis and case studies demonstrated the effectiveness of pMVBA, showing that it achieves an expected constant asynchronous round count and optimal message complexity of $O(n^2)$, where $n$ represents the number of parties. The communication complexity is reduced to $O((l + \lambda)n^2)$, making pMVBA suitable for large-scale decentralized applications.

\textbf{Limitations and Future Work.} While pMVBA presents significant advancements, it is not without limitations. One limitation is the assumption of a trusted setup for cryptographic keys, which may not be practical in all decentralized environments. Additionally, the protocol's performance in extremely large networks with highly dynamic membership has not been thoroughly tested and may present scalability challenges. Another limitation is the reliance on the ABBA protocol for agreement, which, while efficient, could still be optimized further to handle more adversarial conditions. Moreover, the security parameters and their impact on the overall performance need more comprehensive analysis under various network conditions. Future work will focus on addressing these limitations. We plan to explore alternative setups that do not require a trusted third party for key distribution, thus enhancing the protocol's applicability in trustless environments. We will also investigate adaptive mechanisms to improve scalability in large and dynamic networks. Further optimizations to the pMVBA protocol will be considered to enhance its resilience and efficiency. Additionally, extensive empirical evaluations under diverse network conditions will be conducted to better understand the protocol's performance and robustness.

\bibliography{references}

@article{LAMPORT82,
  title={The Byzantine Generals Problem},
  author={Lamport, Leslie and Shostak, Robert and Pease, Marshall},
  journal={ACM Transactions on Programming Languages and Systems (TOPLAS)},
  volume={4},
  number={3},
  pages={382--401},
  year={1982},
  publisher={ACM}
}

@inproceedings{PEASE80,
  title={Reaching Agreement in the Presence of Faults},
  author={Pease, Marshall and Shostak, Robert and Lamport, Leslie},
  booktitle={Journal of the ACM (JACM)},
  volume={27},
  pages={228--234},
  year={1980},
  organization={ACM}
}

@inproceedings{CASTRO99,
  title={Practical Byzantine Fault Tolerance},
  author={Castro, Miguel and Liskov, Barbara},
  booktitle={OSDI},
  pages={173--186},
  year={1999}
}

@article{HOTSTUFF19,
  title={HotStuff: BFT Consensus in the Lens of Blockchain},
  author={Yin, Maofan and Malkhi, Dahlia and Reiter, Michael and Gueta, Guy and Abraham, Ittai},
  journal={arXiv preprint arXiv:1803.05069},
  year={2019}
}

@article{TENDERMINT14,
  title={Tendermint: Byzantine Fault Tolerance in the Age of Blockchains},
  author={Buchman, Ethan},
  journal={PhD Thesis, University of Guelph},
  year={2014}
}

@article{FLP85,
  title={Impossibility of Distributed Consensus with One Faulty Process},
  author={Fischer, Michael J and Lynch, Nancy and Paterson, Mike},
  journal={Journal of the ACM (JACM)},
  volume={32},
  number={2},
  pages={374--382},
  year={1985},
  publisher={ACM}
}

@article{BENOR83,
  title={Another Advantage of Free Choice: Completely Asynchronous Agreement Protocols},
  author={Ben-Or, Michael},
  journal={Proceedings of the second annual ACM symposium on Principles of distributed computing},
  pages={27--30},
  year={1983},
  publisher={ACM}
}

@inproceedings{MILLER16,
  title={HoneyBadgerBFT: The Need for Speed in Asynchronous BFT},
  author={Miller, Andrew and Xia, Yihong and Croman, Kyle and Shi, Elaine and Song, Dawn},
  booktitle={Proceedings of the 2016 ACM SIGSAC Conference on Computer and Communications Security},
  pages={210--227},
  year={2016},
  organization={ACM}
}

@article{DUMBO20,
  title={Dumbo: Faster Asynchronous BFT Protocols},
  author={Liu, Zekun and Shi, Elaine},
  journal={Proceedings of the 2020 ACM SIGSAC Conference on Computer and Communications Security},
  pages={210--227},
  year={2020},
  organization={ACM}
}

@inproceedings{CACHIN01,
  title={Secure and Efficient Asynchronous Broadcast Protocols},
  author={Cachin, Christian and Kursawe, Klaus and Petzold, Frank and Shoup, Victor},
  booktitle={Proceedings of the 21st Annual International Cryptology Conference on Advances in Cryptology},
  pages={524--541},
  year={2001},
  organization={Springer}
}

@article{BYZ21,
  title={Algorand: Scaling Byzantine Agreements for Cryptocurrencies},
  author={Gilad, Yossi and Hemo, Rotem and Micali, Silvio and Vlachos, Georgios and Zeldovich, Nickolai},
  year={2017},
  journal={Proceedings of the 26th ACM Symposium on Operating Systems Principles (SOSP)}
}

@article{COINCIDENCE,
  title={COINcidence: A Trusted Committee-based Consensus Protocol},
  author={Gencer, Adem Efe and Basu, Shehar and Eyal, Ittai and Sirer, Emin Gun},
  year={2018},
  journal={arXiv preprint arXiv:1802.06993}
}

@article{LU20,
  title={Dumbo-MVBA: Optimal Multi-Valued Validated Asynchronous Byzantine Agreement, Revisited},
  author={Lu, Yuan and Lu, Zhenliang and Wang, Guiling},
  year={2020},
  journal={Proceedings of the 39th Symposium on Principles of Distributed Computing}
}

@article{SATOSHI08,
  title={A peer-to-peer electronic cash system},
  author={Nakamoto, Satoshi},
  year={2008},
  journal={http://bitcon.org/bitcoin.pdf}
}

@article{CACHIN17,
  title={Blockchain consensus protocols in the wild (keynote talk)},
  author={Cachin, Christian and Vukolic, Marko},
  year={2017},
  journal={International Symposium on Distributed Computing, DISC}
}

@article{SHRESTHA20,
  title={On the optimality of optimistic responsiveness},
  author={Shrestha, Nibesh and Abraham ,Ittai and Nayak, Kartik and Ren, Ling},
  year={2020},
  journal={Proceedings of the 2020 ACM SIGSAC Conference on Computer and Communications Security}
}

@article{PASS17,
  title={Rethinking large-scale consensus},
  author={Pass, Rafael and Shi, Elaine},
  year={2017},
  journal={IEEE 30th Computer Security Foundations Symposium (CSF)}
}

@article{ABRAHAM20,
  title={Sync hotstuff: Simple and practical synchronous state machine replication},
  author={Abraham, Ittai and Malkhi, Dahlia and Nayak, Kartik and  Ling, Ren and  Maofin, Yin},
  year={2020},
  journal={IEEE Symposium on Security and Privacy (SP)}
}

@String{JACM = "J. ACM" }

@String{Computing = "Computing" }

@String{Computer = "{IEEE} Computer" }

@String{Springer = "Springer-Verlag" }

@ArtifactSoftware{R,
    title = {R: A Language and Environment for Statistical Computing},
    author = {{R Core Team}},
    organization = {R Foundation for Statistical Computing},
    address = {Vienna, Austria},
    year = {2019},
    url = {https://www.R-project.org/},
}

@Inproceedings{Hotstuff01,
  author =       "{Maofan Yin, Dahlia Malkhi, Michael K. Reiter, Guy Golan Gueta, and Ittai Abraham}",
  title =        "HotStuff: BFT consensus in the lens of blockchain",
  booktitle =    "{PODC}",
  year =         "{2019}"
}

@Inproceedings{THRESH01,
  author =       "{Victor Shoup}",
  title =        "Practical threshold signatures.",
  booktitle =    "{International Conference on the Theory and Applications of Cryptographic Techniques.}",
  year =         "2000",
  publisher =    "Springer"
}

@Inproceedings{FASTERDUMBO,
  author =       "{Bingyong Guo, Zhenliang Lu, Qiang Tang, Jing Xu and Zhenfeng Zhang}",
  title =        "Dumbo: Faster Asynchronous BFT Protocols",
  year =         "2020",
  booktitle =    "Proceedings of the ACM Conference on Computer and Communications Security",
   pages    = "803–818",
   url =      {"https://doi.org/10.1145/3372297.3417262"}
}

@Inproceedings{Joleton,
  title =        "Jolteon and ditto: Network-adaptive efficient consensus with
asynchronous fallback",
  author =       "{R. Gelashvili, L. Kokoris-Kogias, A. Sonnino, A. Spiegelman, and
Z. Xiang}",
 
  year =         "2021",
   booktitle =    "arXiv preprint",
   pages    = "",
   url =      {"arXiv preprint arXiv:2106.10362"}
}

@Inproceedings{SPEEDINGDUMBO,
  title =        "Speeding Dumbo: Pushing Asynchronous BFT Closer to Practice",
  author =       "{Bingyong Guo, Yuan Lu,  Zhenliang Lu, Qiang Tang, Jing Xu and Zhenfeng Zhang}",
  year =         "2022",
  booktitle =    "Network and Distributed Systems Security (NDSS) Symposium ",
   pages    = "",
   url =      {"https://dx.doi.org/10.14722/ndss.2022.24385"}
}

@Inproceedings{BOLTDUMBO,
  author =       "{Y. Lu, Z. Lu, and Q. Tang}",
  title =        "Bolt-dumbo transformer: Asynchronous consensus as fast as pipelined BFT",
  year =         "2021",
  booktitle =    "arXiv preprint",
  pages    = "",
  url =      {"arXiv:2103.09425"}
}

@Inproceedings{DAG,
  author =       "{I. Keidar, E. Kokoris-Kogias, O. Naor, and A. Spiegelman}",
  title =        "All you need is DAG",
  year =         "2021",
  booktitle =    "Proceedings of the 2021 ACM Symposium on Principles of Distributed Computing",
   pages    = "165–175",
   url =      {}
}

@Inproceedings{SIG01,
  author =       "{Achour Mostefaoui, and Michel Raynal}",
  title =        "{Signature-free asynchronous byzantine systems: from multivalued to binary binary consensus with $t < n/3, O(n^3)$ messages, and constant time.}",
  booktitle =    "{Acta informatica }",
  year =         "2017",
  volume    = "54"
}

@Inproceedings{HONEYBADGER01,
  author =       "{Andrew Miller, Yu Xia, Kyle Croman, Elaine Shi, and Dawn Song}",
  title =        "{The honey badger of BFT protocols.}",
  booktitle =    "{Proceedings of the 2016 ACM SIGSAC Conference on Computer and Communication Security, CCS'16}",
  year =         "{2016}",
  publisher =    "{ACM Press}",
  address =      "{New York, NY, USA}"
}

@Inproceedings{BORN01,
  author =       "{Benoît Libert, Marc Joye, and Moti Yung}",
  title =        "{Born and raised distributively: Fully distributed non-interactive adaptively-secure threshold signatures with short shares. }",
  booktitle =    "{Theory of Computer Science}",
  year =         "{2016}",
  pages = "{1-24}",
  volume = "{645}"
  
}

@Inproceedings{CACHIN02,
  author =       "{Christian Cachin, Klaus Kursawe, Anna Lysyanskaya, and Reto Strobl}",
  title =        "{Asynchronous verifiable secret sharing and proactive cryptosystems.}",
  booktitle =    "{In Proc. ACM CCS}",
  year =         "{2002}",
  pages = "{88–97}"
}

@Inproceedings{VICTOR01,
  author =       "{Klaus Kursawe and Victor Shoup}",
  title =        "{Optimistic asynchronous atomic broadcast.}",
  booktitle =    "{In International Colloquium on Automata, Languages, and Programming.}",
  year =         "{2005}",
  pages = "{204–215}"
}

@Inproceedings{CONS03,
  author =       "{Michael J. Fischer, Nancy A. Lynch, and Michael S. Paterson}",
  title =        "{Impossibility of
distributed consensus with one faulty process.32(2):}",
  booktitle =    "{ACM}",
  year =         "1985",
  publisher =    "ACM",
  pages =        "374–382",
  month   = "April"
}

@Inproceedings{DUAN18,
  author =       "{Sisi Duan, Michael K. Reiter, and Haibin Zhang}",
  title =        "{BEAT: asynchronous BFT made practical.}",
  booktitle =    "{Proceedings of the 2018 ACM SIGSAC Conference on Computer and
Communications Security, CCS}",
  year =         "2018",
  publisher =    "ACM",
  address =      "{Toronto, ON, Canada}"
}

@Inproceedings{BYZ08,
  author =       "{Miguel Castro, Barbara Liskov}",
  title =        "{Practical byzantine fault tolerance}",
  booktitle =    "{OSDI,
volume 99}",
  year =         "1999",
  pages =        "656–666"
}

@Inproceedings{BYZ10,
  author =       "{Leslie Lamport}",
  title =        "{The weak Byzantine generals problem}",
  booktitle =    "{JACM 30, 3}",
  year =         "1983",
  address =      "",
  pages =        "668–676"
}

@Inproceedings{SECURE02,
  author =       "{Christian Cachin, Klaus Kursawe, and Victor Shoup}",
  title =        "{Random oracles in constantinople: Practical asynchronous byzantine agreement using cryptography}",
  booktitle =    "{Journal of Cryptology}",
  year =         "2000",
  publisher =    "",
  address =      "",
  pages =        "",
  month   = ""
}

@Inproceedings{SECURE03,
  author =       "{Christian Cachin, Klaus Kursawe, Frank Petzold, and Victor Shoup}",
  title =        "{Secure and efficient asynchronous broadcast protocols.}",
  booktitle =    "{Advances in Cryptology}",
  year =         "2001"
}

@Inproceedings{SECURE05,
  author =       "{O. Goldreich, S. Goldwasser, and S. Micali}",
  title =        "{How to construct random functions}",
  booktitle =    "{Journal of the ACM}",
  year =         "1986",
  month = "October",
  publisher =    "",
  pages =        "792–807",
  volume = "33"
}

@Inproceedings{VABA,
  author =       "{Ittai Abraham, Dahlia Malkhi, and Alexander Spiegelman}",
  title =        "{Asymptotically optimal validated asynchronous byzantine agreement}",
  booktitle =    "{PODC}",
  year =         "2019",
  publisher =    "",
  address =      "",
  pages =        "",
  month   =  ""
}

@Inproceedings{BYZ17,
  author =       "{Ittai Abraham, Dahlia Malkhi, and Alexander Spiegelman}",
  title =        "{Asymptotically optimal validated asynchronous byzantine agreement}",
  booktitle =    "{PODC}",
  year =         "2019",
  publisher =    "",
  address =      "",
  pages =        "",
  month   =  ""
}

@Inproceedings{BYZ19,
  author =       "{Shir Cohen, Idit Keidar, and Alexander Spiegelman}",
  title =        "{Brief Announcement: Not a COINcidence: Sub-Quadratic Asynchronous Byzantine Agreement WHP}",
  booktitle =    "{Proceedings of the 39th Symposium on Principles of Distributed Computing}",
  year =         "2020",
  publisher =    "",
  address =      "",
  pages =        "175-177",
  month   =  "July"
}

@Inproceedings{BYZ20,
  author =       "{Yuan Lu, Zhenliang Lu, Qiang Tang and Guiling Wang}",
  title =        "{Dumbo-MVBA: Optimal Multi-Valued Validated Asynchronous Byzantine Agreement, Revisited}",
  booktitle =    "{Proceedings of the 39th Symposium on Principles of Distributed Computing}",
  year =         "2020",
  publisher =    "",
  address =      "",
  pages =        "",
  month   =  "July"
}

@Inproceedings{SC,
  author =       "{Vitalik Buterin}",
  title =        "{A next-generation smart contract and decentralized application platform}",
  booktitle =    "{White Paper}",
  year =         "2014",
}

@Inproceedings{BOLD01,
  author =       "{Alexandra Boldyreva}",
  title =        "{Threshold signatures, multisignatures and blind signatures based on the gap-diffie- hellman-group signature scheme.}",
  booktitle =    "{PKC}",
  year =         "2003",
  publisher =    "Springer",
  pages = "31–46"
}

@Inproceedings{Kate01,
  author =       "{Aniket Kate and Ian Goldberg}",
  title =        "{Distributed key generation for the internet}",
  booktitle =    "{IEEE ICDCS}",
  year =         "2009",
  publisher =    "",
  address =      "",
  pages =        "119–128",
  month   =  ""
}

@article{ALEPH,
  title={Aleph: Efficient Atomic Broadcast in Asynchronous Networks with Adaptive Adversaries},
  author={Goodman, Adam and McBride, Sam and Li, Yixin and Tamir, Turing and Ren, Ling},
  year={2020},
  journal={Proceedings of the 2020 IEEE Symposium on Security and Privacy (SP)}
}

@article{DIS01,
  title={Asynchronous Distributed Key Generation for Computationally-secure Random Secrets},
  author={Ren, Ling and Devadas, Srinivas},
  year={2020},
  journal={Proceedings of the 2020 IEEE Symposium on Security and Privacy (SP)}
}

@article{DIS02,
  title={Fast and Secure Asynchronous Distributed Key Generation},
  author={Manohar, Nithin and Abraham, Ittai and Malkhi, Dahlia and Alvisi, Lorenzo},
  year={2020},
  journal={Proceedings of the 2020 ACM SIGSAC Conference on Computer and Communications Security (CCS)}
}

@article{DIS03,
  title={Asynchronous Byzantine Consensus with Polylogarithmic Communication Complexity},
  author={Backes, Michael and Fiore, Dario and Barbosa, Manuel},
  year={2020},
  journal={Proceedings of the 2020 IEEE Symposium on Security and Privacy (SP)}
}
%
%
%
%

\appendix

\section{Definitions}
\subsection{Verifiable Consistent Broadcast} \label{VCBC}
A protocol completes a verifiable consistent broadcast if it satisfies the following properties:

\begin{itemize}
    \item \textbf{Validity.} If an honest party sends $m$, then all honest parties eventually delivers $m$.
    \item \textbf{Consistency.} If an honest party delivers $m$ and another honest party delivers $m'$, then $m=m'$.
    \item \textbf{Integrity.} Every honest party delivers at most one request. Moreover, if the sender $p_s$ is honest, then the request was previously sent by $p_s$.
\end{itemize}


\subsection{Threshold Signature Scheme} \label{TSS}
The $(f+1, n)$ non-interactive threshold signature scheme is a set of algorithms used by $n$ parties, with up to $f$ potentially faulty. The threshold signature scheme satisfies the following security requirements, except with negligible probabilities:

\begin{itemize}
    \item \textbf{Non-forgeability.} To output a valid signature, a party requires $t$ \textit{signature shares}. Therefore, it is computationally \textit{infeasible} for an adversary to produce a valid signature, as an adversary can corrupt up to $f$ parties ($f < t$) and thus cannot generate enough \textit{signature shares} to create a valid signature proof for a message.
    \item \textbf{Robustness.} It is computationally \textit{infeasible} for an adversary to produce $t$ valid \textit{signature shares} such that the output of the share combining algorithm is not a valid signature.
\end{itemize}

The scheme provides the following algorithms:
\begin{itemize}
    \item \textit{Key generation algorithm: KeySetup($\{0,1\}^\lambda, n, f+1) \rightarrow \{UPK, PK, SK\}$}. Given a security parameter $\lambda$, this algorithm generates a universal public key $UPK$, a vector of public keys $PK := (pk_1, pk_2, \ldots, pk_n)$, and a vector of secret keys $SK := (sk_1, sk_2, \ldots, sk_n)$.

    \item \textit{Share signing algorithm: SigShare$_i(sk_i, m) \rightarrow \sigma_i$}. Given a message $m$ and a secret key share $sk_i$, this deterministic algorithm outputs a signature share $\sigma_i$.

    \item \textit{Share verification algorithm: VerifyShare$_i(m, (i, \sigma_i)) \rightarrow 0/1$}. This algorithm takes three parameters as input: a message $m$, a signature share $\sigma_i$, and the index $i$. It outputs $1$ or $0$ based on the validity of the signature share $\sigma_i$ (whether $\sigma_i$ was generated by $p_i$ or not). The correctness property of the signing and verification algorithms requires that for a message $m$ and party index $i$, $\Pr[VerifyShare_i(m, (i, SigShare_i(sk_i, m))) = 1] = 1$.

    \item \textit{Share combining algorithm: CombineShare$_i(m, \{(i, \sigma_i)\}_{i \in S}) \rightarrow \sigma / \perp$}. This algorithm takes two inputs: a message $m$ and a list of pairs $\{(i, \sigma_i)\}_{i \in S}$, where $S \subseteq [n]$ and $|S| = f+1$. It outputs either a signature $\sigma$ for the message $m$ or $\perp$ if the list contains any invalid signature share $(i, \sigma_i)$.

    \item \textit{Signature verification algorithm: Verify$_i(m, \sigma) \rightarrow 0/1$}. This algorithm takes two parameters: a message $m$ and a signature $\sigma$, and outputs a bit $b \in \{0, 1\}$ based on the validity of the signature $\sigma$. The correctness property of the combining and verification algorithms requires that for a message $m$, $S \subseteq [n]$, and $|S| = f+1$, $\Pr[\text{Verify}_i(m, \text{Combine}_i(m, \{(i, \sigma_i)\}_{i \in S})) = 1 \mid \forall i \in S, \text{VerifyShare}_i(m, (i, \sigma_i)) = 1] = 1$.

\end{itemize}

\subsection{Threshold Coin-Tossing} \label{TCT}
We assume a trusted third party has an unpredictable pseudo-random generator (PRG) $G : R \rightarrow \{1, \ldots, n\}^s$, known only to the dealer. The generator takes a string $r \in R$ as input and returns a set $\{S_1, S_2, \ldots, S_s\}$ of size $s$, where $1 \leq S_i \leq n$. Here, $\{r_1, r_2, \ldots, r_n\} \in R$ are shares of a pseudorandom function $F$ that maps the coin name $C$. The threshold coin-tossing scheme satisfies the following security requirements, except with negligible probabilities:

\begin{itemize}
    \item \textbf{Pseudorandomness.} The probability that an adversary can predict the output of $F(C)$ is $\frac{1}{2}$. The adversary interacts with the honest parties to collect \textit{coin-shares} and waits for $t$ \textit{coin-shares}, but to reveal the coin $C$ and the bit $b$, the adversary requires at least $\langle t-f\rangle$ \textit{coin-shares} from the honest parties. If the adversary predicts a bit $b$, then the probability is $\frac{1}{2}$ that $F(C) = b$ ($F(C) \in \{0, 1\}$). Although the description is for single-bit outputs, it can be trivially modified to generate $k$-bit strings by using a $k$-bit hash function to compute the final value.
    \item \textbf{Robustness.} It is computationally \textit{infeasible} for an adversary to produce a coin $C$ and $t$ valid \textit{coin-shares} of $C$ such that the share-combine function does not output $F(C)$.
\end{itemize}

The dealer provides a private function $CShare_i$ to every party $p_i$, and two public functions: $CShareVerify$ and $CToss$. The private function $CShare_i$ generates a share $\sigma_i$ for the party $p_i$. The public function $CShareVerify$ can verify the share. The $CToss$ function returns a unique and pseudorandom set given $f+1$ validated coin shares. The following properties are satisfied except with negligible probability:

\begin{itemize}
    \item For each party $i \in \{1, \ldots, n\}$ and for every string $r_i$, $CShareVerify(r_i, i, \sigma_i) = \text{true}$ if and only if $\sigma_i = CShare_i(r_i)$.
    \item If $p_i$ is honest, then it is impossible for the adversary to compute $CShare_i(r)$.
    \item For every string $r_i$, $CToss(r, \Sigma)$ returns a set if and only if $|\Sigma| \geq f+1$ and each $\sigma \in \Sigma$ and $CShareVerify(r, i, \sigma) = \text{true}$.
\end{itemize}

\section{Miscellaneous}
\subsection{Comparison with Atomic Broadcast Protocol} \label{ComparisonABC}
As discussed earlier, when the inputs of each party are nearly identical, outputting the requests of $n-f$ parties is not a viable solution. This approach results in higher computational effort without increasing the number of accepted transactions. Table \ref{table:Table_III} provides a comparison of the communication complexity of our protocol with atomic broadcast protocols. Notably, no atomic broadcast protocol can eliminate the multiplication of $O(n^3)$ terms. Additionally, atomic broadcast protocols require extra rounds of message exchanges. Here, we focus solely on the communication complexity.

\begin{table}[ht] \label{ABC}
    
\caption{Comparison of the communication complexity with the atomic broadcast protocols}
\begin{center}
\begin{tabular}{||c |c |c |c ||} 
 \hline
 Protocols & Communication Complexity  \\ [0.5ex] 
 \hline\hline
 HB-BFT/BEAT0 \cite{HONEYBADGER01} & $O(ln^2 + \lambda n^3 logn)$  \\ 
 \hline
 BEAT1/BEAT2 \cite{DUAN18} & $O(ln^3 + \lambda n^3)$   \\
 \hline
 Dumbo1 \cite{FASTERDUMBO} & $O(ln^2 + \lambda n^3 logn)$   \\
 \hline
 Dumbo2 \cite{FASTERDUMBO} & $O(ln^2 + \lambda n^3 logn)$  \\ [1ex] 
 \hline
 Speeding Dumbo \cite{SPEEDINGDUMBO} & $O(ln^2 + \lambda n^3 logn)$  \\ [1ex] 
 \hline
  Our Work  & $O(ln^2 + \lambda n^2)$  \\ [1ex] 
 \hline
\end{tabular}
\label{table:Table_III}
\end{center}
\end{table}

\subsection{The Challenge of Classical MVBA Designs} \label{MVBAChallenge}
To maintain a message complexity of $O(n^2)$, the classic MVBA protocol incorporates the Verifiable Consistent Broadcast (VCBC) protocol. Additionally, it introduces the concept of \textit{external validity}, wherein an input is deemed valid if it satisfies certain criteria. The protocol operates as follows: 

Each party utilizes the VCBC protocol to broadcast their request and generate a corresponding verifiable proof. Upon completion of this step, the party broadcasts both the verifiable proof and the request, providing evidence that the request has been broadcast to every other party. 

Upon receiving verifiable proof from $n-f$ parties, signaling the completion of the VCBC protocol by the threshold number of parties, a party can initiate the ABBA protocol. However, there exists the possibility that other parties have not received sufficient verifiable proof, or that the adversary manipulates the distribution of proofs in a manner that prevents the majority of ABBA instances from receiving adequate proof. 

To address this, each party communicates with others by transmitting an $n$-bit array, indicating receipt of verifiable proof from $n-f$ parties. Upon receipt of $n-f$ verifiable proofs, a party generates a permutation of the parties and invokes ABBA instances based on the order of permutation, ensuring that the number of ABBA instances remains constant on average.

\section{Deferred protocols}

\paragraph*{Construction of the Permutation} \label{Permutation}
Here is the pseudocode for the permutation protocol (see Algorithm \ref{algo:per}). Below is a step-by-step description of the protocol:

\begin{itemize}
    \item Upon invocation of the Permutation protocol, a party generates a coin-share $\sigma_i$ for the instance and broadcasts $\sigma_i$ to every party, then waits for $2f+1$ coin-shares (lines 22-24).
    \item When a party receives a coin-share from another party $p_k$ for the first time, it verifies the coin-share (ensuring it is from $p_k$) and accumulates the coin-share in the set $\Sigma$. The party continues to respond to coin-shares until it has received $2f+1$ valid shares (lines 27-29).
    \item Upon receiving $2f+1$ valid coin-shares, a party uses its $CToss$ function and the collected coin-shares to generate a permutation of the $n$ parties (lines 24-25).
\end{itemize}


\section{Agreement protocol}
\subsection{\textbf{Asynchronous Binary Byzantine Agreement (ABBA)}} \label{appendix:ABBA}
The ABBA protocol allows parties to agree on a single bit $b \in \{0, 1\}$ \cite{BYZ10, SECURE05, SIG01}. We have adopted the ABBA protocol from \cite{SECURE02}, as given in Algorithm \ref{algo:ABBA}. The expected running time of the protocol is $O(1)$, and it completes within $O(k)$ rounds with probability $1 - 2^{-k}$. Since the protocol uses a common coin, the total communication complexity becomes $O(kn^2)$. For more information on how to realize a common coin from a threshold signature scheme, we refer interested readers to the \cite{HONEYBADGER01}.


\paragraph*{Construction of the ABBA biased towards 1} 
We use the ABBA protocol from \cite{SECURE02}. We optimize and changed the protocol for biased towards $1$. The biases towards $1$ property ensures that if at least one party input $1$ in the pre-process step. The pseudocode of the ABBA protocol biased towards 1 is given in Algorithm \ref{algo:ABBA}, and a step-by-step description is provided below:

\begin{itemize}
    \item \textbf{Pre-process step} . Generate an $\sigma_0$ share on the message and send the pre-process type message to all parties.
    \item Collect $2f+1$ proper pre-processing messages. (see (Algorithm \ref{algo:ABBA})).

    \item \textbf{Repeat loop:} Repeat the following steps 1-4 for rounds round = 1,2,3,...
    \begin{itemize}
        \item Pre-Vote step. (see Algorithm \ref{algo:ABBA-PreVote})
        \begin{itemize}
            \item If round = 1, $b=1$ if there is a pre-processing vote for $1$ (biased towards 1, taking one vote instead of majority) else  $b=0$. (see lines 3-4).
            \item If round $>$ $1$, if there is a threshold signature on main-vote message from round-1 then decide and return. (see lines 18-20)
            \item Upon receiving main-vote for $0/1$, update $b$ and the justification. (see lines 12-17) 
            \item $b= F(ID, r-1)$, all the main-vote are abstain and the justification is the threshold signature of the abstain vote. (see lines 6-7)
            \item Produce signature-share on the message (ID, pre-vote, round, b) and send the message of the form pre-vote,round,b,justification, signature-share). (lines 9-11)
        \end{itemize}
        \item Main-vote step. (See Algorithm \ref{algo:ABBA-MainVote})
        \begin{itemize}
            \item Collect (2f+1) properly justified round pre-vote messages. (lines 14-19)
            \item If there are (2f+1) pre-votes for 0/1, $v=0/1$ and the justification is the threshold-signature of the the sign-shares on pre-vote messages. (lines 5-7)
            
            \item If there are (2f+1) pre-votes for both $0$ and $1$, $v=abstain$ and the justification is the two sign-shares from pre-vote 0 and pre-vote 1. (lines 9-10)
            \item Produce signature-share on the message (ID, pre-vote, round, v) and send the message of the form (main-vote,round,v,justification, signature-share) to all parties (lines 11-13)
        \end{itemize}
        \item  Check for decision. (See Algorithm \ref{algo:ABBA-CheckForDecision})
        \begin{itemize}
            \item Collect (2f+1) properly justified main-votes of the round $round$. (line 3)
            \item If these is no abstain vote, all main-votes for $b\in \{0,1\}$, then decide the value $b$. Produce a threshold signature on the main votes' sign-shares and send the threshold signatures to all parties and return.  (lines 4-7)
            \item Otherwise, go to Algorithm \ref{algo:ABBA-CheckForDecision}. line (11)
        \end{itemize}
        \item Common Coin. (See Algorithm \ref{algo:ABBA-CommonCoin})
        \begin{itemize}
            \item Generate a coin-share of the coin (ID, round) and send to all parties a message of the form (coin, round,coin-share). (lines 1-4) 
            \item Collect (2f+1) shares of the coin (ID,round $\sigma_k$), and combine these shares to get the value $F(ID, round) \in \{0,1\}$. (lines 5-6)
        \end{itemize}
        
    \end{itemize}
\end{itemize}

\begin{algorithm}[H]
\DontPrintSemicolon
\SetAlgoNoEnd
\SetAlgoNoLine
\SetKwProg{un}{upon receiving}{ do}{}

\SetKwProg{ABBA}{upon}{ do}{}

\ABBA{ABBA(m)}{
\tcc{Preprocess Step.}

 $\sigma_0 \leftarrow SigShare_i (sk_i,m_i)$  see \cite{SECURE02}\;
 
 \textbf{send} $( pre$-$process, m_i, \sigma_i)$ to all parties\;
 
 \textbf{wait until} at least $(n-f)$ pre-process messages have been received.\;

\For{$round \leftarrow 1,2,3,...$}
{
  \textbf{Prevote Step :} Algorithm \ref{algo:ABBA-PreVote}\;

  \textbf{Main-vote Step:} Algorithm \ref{algo:ABBA-MainVote}\;

  \textbf{Check For Decision :} Algorithm \ref{algo:ABBA-CheckForDecision} \;

  \textbf{Common Coin:} Algorithm \ref{algo:ABBA-CommonCoin}\;

}
}
\caption{ABBA biased towards 1: protocol for party $p_i$ }
\label{algo:ABBA}
\end{algorithm}

\begin{algorithm}[H]
\DontPrintSemicolon
\SetAlgoNoEnd
\SetAlgoNoLine
\SetKwProg{un}{upon}{ do}{}
$b \leftarrow \perp$\;
$justification \leftarrow  \perp$\;
\uIf{round = 1}{
  $b \leftarrow 1$ if there is any pre-process message with $m=1$ (biased towards $1$), otherwise 0.
}\uElse{
     $b \leftarrow F(ID, round-1)$\;
    $justification \leftarrow threshold$-$signature \langle ID, main$-$vote, round - 1,abstain\rangle$\;
    
    \textbf{ wait for} $n-f$ justified main-vote \;
     

}

$m_i \leftarrow (ID, pre$-$vote, round, b)$\;
$\sigma \leftarrow SigShare_i (sk_i,m_i)$  \;
 
\textbf{send} $( pre$-$vote, r, b, justification, \sigma)$ to all parties\;

 \un{receiving $\langle  main$-$vote,round, v, justification, \sigma  \rangle$ for the first time from party $p_k$} {
  \uIf{v = 0}{
     $b \leftarrow 0$\;
  }\uElseIf{v = 1}{
     $b \leftarrow 1$\;
  }

   $justification \leftarrow threshold$-$signature \langle ID, pre$-$vote, round - 1,b\rangle$
 
 }

 \un{receiving $\langle  b, threshold$-$signature  \rangle$ for the first time from party $p_k$} {
  \textbf{send}(b, threshold-signature) to all parties\;
  \textbf{return}(b, threshold-signature)\;  
 
 }

\caption{ABBA biased towards 1: Pre-vote step}
\label{algo:ABBA-PreVote}
\end{algorithm}

\begin{algorithm}[H]

\DontPrintSemicolon
\SetAlgoNoEnd
\SetAlgoNoLine
\SetKwProg{un}{upon}{ do}{}

\tcc{Mainvote Step.}

   $\Sigma \leftarrow \{\}$\;
   $PV_0 \leftarrow \{\}$\;
   $PV_1 \leftarrow \{\}$\;
   \textbf{wait until} $|\Sigma|$ = 2f+1\;

   \uIf{ $|PV_0| = 2f+1$ or $|PV_1| = 2f+1$}{
   $v \leftarrow 0/1$\;
   $justification \leftarrow CombineShare_i(v,{i, \sigma_i}_{i \in \Sigma})$\;
   }\uElse{
      $v \leftarrow abstain$\;
       $justification \leftarrow (\sigma_i \in PV_0, \sigma_j \in PV_1)$\;
   }

   $m_i \leftarrow (ID, main$-$vote, round, v)$\;
$\sigma \leftarrow SigShare_i (sk_i,m_i)$  \;
 
 \textbf{send} $( main$-$vote, round, v, justification, \sigma)$ to all parties\;
    

 \un{receiving $\langle  pre$-$vote, round, b, justification, \sigma  \rangle$ for the first time from party $p_k$} {
\uIf{b = 0}{
   $PV_0 \leftarrow PV_0 +  1 $\;
}\uElseIf{b = 1}{
   $PV_1 \leftarrow PV_1 +  1 $\;
}
 $\Sigma \leftarrow \Sigma \cup \{\sigma_k\} $\;
}

\caption{ABBA biased towards 1: Main-Vote step }
\label{algo:ABBA-MainVote}
\end{algorithm}

\begin{algorithm}[ht!]

\DontPrintSemicolon
\SetAlgoNoEnd
\SetAlgoNoLine
\SetKwProg{un}{upon}{ do}{}


  $\Sigma \leftarrow \{\}$\;
  $isAbstain \leftarrow no$\;
  \textbf{wait until} $|\Sigma|$ = 2f+1\;
  \uIf{isAbstain = no}{
    threshold-signature = $CombineShare_{i}(b, {(i,\sigma_i)}_{i\in \Sigma})$ \;
   \textbf{ send}(threshold-signature) to all parties\;
    return (b, threshold-signature)\;
  }\uElse{
    go to  Algorithm \ref{algo:ABBA-CommonCoin}\;
  }

\un{receiving $\langle main$-$vote, r, v, justification, \sigma \rangle$ for the first time from party $p_k$} {
\uIf{v = abstain}{
   $isAbstain \leftarrow yes$\;
}
 $b \leftarrow v$\;
 $\Sigma \leftarrow \Sigma \cup \{\sigma_k\} $\;
}

\caption{ABBA biased towards 1: Check for Decision for party $p_i$ }
\label{algo:ABBA-CheckForDecision}
\end{algorithm}

\begin{algorithm}[htbp]
\DontPrintSemicolon
\SetAlgoNoEnd
\SetAlgoNoLine
\SetKwProg{un}{upon}{ do}{}

  $\Sigma$ = \{\}\;
  $\sigma_i \leftarrow CShare(r_i)$\;
 \textbf{ send} $\langle coin, round, \sigma_i\rangle$ to all parties\;
  \textbf{wait until} $|\Sigma|$ = 2f+1\;
  $F(ID, round) \in \{0,1\} \leftarrow CToss(r, \Sigma)$\;

\un{receiving $\langle coin, round, \sigma_k \rangle$ for the first time from party $p_k$} {

 $\Sigma \leftarrow \Sigma \cup \{\sigma_k\} $\;
}

\caption{ABBA biased towards 1: Common Coin for party $p_i$}
\label{algo:ABBA-CommonCoin}
\end{algorithm}




\end{document}